\documentclass[11pt]{article}
\usepackage{algorithm,algorithmicx,algpseudocode,amsmath,amssymb,amsthm}
\usepackage{bm,color,dsfont,epsfig,fullpage,graphicx,pifont}
\usepackage{makecell}
\usepackage[toc,page]{appendix}
\usepackage{url}
\usepackage{multirow}
\usepackage{comment}
\usepackage{enumitem}
\usepackage{xcolor}
\usepackage{epstopdf}
\usepackage{subfig}
\graphicspath{{figures/}}
\usepackage[noabbrev]{cleveref}
\crefname{equation}{}{}

\allowdisplaybreaks

\newtheorem{lemma}{Lemma}
\newtheorem{theorem}{Theorem}

\def\lb{\left(}
\def\rb{\right)}
\def\lcb{\left\{}
\def\rcb{\right\}}
\def\ln{\left\|}
\def\rn{\right\|}
\def\lsb{\left[}
\def\rsb{\right]}
\def\P{\mathcal{P}}
\def\Po{\mathcal{P}_{\Omega}}
\def\Ho{\mathcal{H}_{\Omega}}
\def\I{\mathcal{I}}
\def\T{\mathcal{T}}
\def\De{\Delta}
\def\La{\Lambda}
\def\Si{\Sigma}
\DeclareMathOperator{\rank}{rank}
\DeclareMathOperator{\diag}{diag}

\begin{document}

\title{Leave-One-Out Analysis for  Nonconvex Robust Matrix Completion with General Thresholding Functions}

\author{Tianming Wang\thanks{T. Wang is with the School of Mathematics, Southwestern University of Finance and Economics, Chengdu, Sichuan, China (wangtm@swufe.edu.cn).} \and Ke Wei\thanks{K. Wei is with the School of Data Science, Fudan University, Shanghai, China (kewei@fudan.edu.cn).}
}

\maketitle

\begin{abstract}
We study the problem of robust matrix completion (RMC), where the partially observed entries of an underlying low-rank matrix is corrupted by  sparse noise. Existing analysis of the non-convex methods for this problem either requires the explicit but empirically redundant regularization in the algorithm or requires sample splitting in the analysis. In this paper, we consider a simple yet efficient nonconvex method which alternates between a projected gradient step for the low-rank part and a thresholding step for the sparse noise part. Inspired by  leave-one-out analysis for low rank matrix completion, it is established that the method can achieve linear convergence for a general class of thresholding functions, including for example soft-thresholding and SCAD. To the best of our knowledge, this is the first leave-one-out analysis on a nonconvex method for RMC. Additionally,
when applying our result to low rank matrix completion, it improves the sampling complexity of existing result for the singular value projection method.
\end{abstract}

\section{Introduction}

Low rank matrix completion has received a lot of investigations in the last decade, both from theoretical
and algorithmic aspects. In the pioneering works of Cand\`es, Recht and Gross \cite{Candes2009, gross2011recovering}, it is shown that the nuclear norm based convex relaxation can achieve exact matrix completion under nearly optimal sampling complexity. 
When the size of the matrix is large, nonconvex methods are more suitable for low rank matrix completion. Many computationally efficient approaches have been developed towards this line of research, including gradient descent \cite{zheng-pg,procrustesflow2016}, iterative hard thresholding \cite{Jain2010,tanner2013normalized}, Riemannian optimization \cite{mishra2012riemannian,vandereycken2013low}, to name just a few. 

The nonconvex methods also merit theoretical recovery guarantees under the notion of incoherence and the standard random sampling assumption. In order to exploit the incoherence property sufficiently, either explicit projection onto the incoherence region \cite{zheng-pg} or sample splitting \cite{Hardt2014,Jain2015,Wei2020} can be utilized in the convergence analysis of the nonconvex methods. 
The latter one splits the observed samples into disjoint subsets in order to create independence across iterations. However, it is well recognized that neither of the two techniques is necessary in practice. 
This gap has recently been filled based on the leave-one-out analysis technique \cite{El2013,El2018,Abbe2020,Adel2018,Zhong2018}. When applying this technique to analyze nonconvex algorithms, the idea is to  build auxiliary sequences that are exceedingly close to each other and also allows the decoupling of the statistical dependency. As a result, the utilization  of leave-one-out  allows us to show that the iterates  always stay in a incoherence region and thus has successfully improved the theoretical guarantees of nonconvex methods for low rank matrix completion. In particular, Ma et al. \cite{Ma2019} establish the linear convergence of the vanilla gradient descent method for low rank matrix completion, while Ding and Chen \cite{Ding2020} establish the linear convergence of the  singular value projection (SVP, also known as iterative hard thresholding) algorithm \cite{Jain2010,tanner2013normalized}.

In this paper, we consider the robust matrix completion (RMC) problem which arises for example in robust PCA from partial samples \cite{Cherapanamjeri2017}. The goal in RMC is to recover the ground truth low-rank matrix $L^{\star}\in\mathbb{R}^{n_1\times n_2}$ from partially observed  measurements that are corrupted by sparse outliers. More precisely, each observation can be written as
\begin{align}
M_{ij} = L^{\star}_{ij} + S^{\star}_{ij},\quad (i,j)\in \Omega,\label{eq:RMC-model}
\end{align}
where  $\Omega\subseteq[n_1]\times[n_2]$ is an index subset with $[n]:=[1,2,\dots,n]$, and $S^{\star}_{ij}$ denotes the outlier. There are also different nonconvex methods developed for RMC. 
In \cite{Yi2016}, a sorting operation is first used to update the outlier estimation, and then  a gradient step is taken to update the low-rank part represented via Burer-Monteiro factorization. The analysis therein  requires the explicit projection  onto the incoherence region which, as in the matrix completion case, is often deemed as an artifact of the proof. In a later work \cite{Cherapanamjeri2017}, an algorithm which combines SVP with hard-thresholding  is considered to solve the RMC problem. In every iteration, it performs a projected gradient step for the low-rank matrix, and then uses hard-thresholding for the outlier estimation. However, the theoretical guarantee in \cite{Cherapanamjeri2017} is derived based on the sample splitting trick. An extension of the leave-one-out analysis in \cite{Ding2020} turns out to be quite challenging when there exist sparse outliers.

 \subsection{Main Contributions}
To the best of our knowledge, there is no leave-one-out analysis for nonconvex methods targeting the RMC problem. In a closely related work \cite{Chen2021}, theoretical recovery guarantee is improved for a convex program by building a connection between the optimal solution of the convex program with the iterates of a nonconvex algorithm. The nonconvex algorithm can be viewed as a  projection-free version of the method used in \cite{Yi2016}, except that the sorting operation is replaced by the soft-thresholding operation. However, the nonconvex algorithm  in \cite{Chen2021} is not a practical one, since it requires initialization at the ground truth and is only used to verify certain optimal conditions. That being said, it motivates us to derive the theoretical guarantee of certain nonconvex methods for RMC based on the leave-one-out analysis technique, thus achieving a projection and sample splitting free result. Overall, the novelty of this paper are  summarized as follows.
\begin{enumerate}
    \item We consider a nonconvex method which alternates between SVP  for the low rank estimation and entrywise thresholding for the sparse outlier estimation. The theoretical recovery guarantee of the nonconvex method has been established based on the leave-one-out analysis technique for a class of thresholding functions, including for example soft-thresholding
and SCAD. It is also worth pointing out that our analysis can handle any sparse enough outliers, without requiring the support and signs of the outliers to be random, as assumed in \cite{Chen2021}. To the best of our knowledge, this is the first leave-one-out analysis for nonconvex methods targeting  the RMC problem. 
    \item 
     When applying our result to the matrix completion problem, our sample complexity  improves substantially upon the one  in \cite{Ding2020} over the dependency on $\kappa$, $\mu$ and $r$, i.e., from $O(\frac{\kappa^6\mu^4r^6\log n}{n})$ to $O(\frac{\kappa^4\mu^3r^3\log n}{n})$.
\end{enumerate}
For the analysis, as inspired by \cite{Ding2020}, we rewrite the update of the algorithm as a perturbation around the ground truth and then establish the contraction of $\|E^{t-1,\infty}\|_2$,  $\|\Delta^{t,\infty}\|_{2,\infty}$, and $\|D^{t,\infty}\|_F$ (see Section~\ref{proof_outline} for the definitions) by induction and the leave-one-out analysis. That being said, we would like to emphasize that the proof details are quite technically involved and new techniques have to be developed to achieve the improved analysis and to handle the outliers.
\begin{itemize}
    \item Firstly, to achieve the improved sample complexity, instead of first bounding $\|\Delta^{t,\infty}\|_{2,\infty}$ and then bounding $\|D^{t,\infty}\|_F$ based on $\|\Delta^{t,\infty}\|_{2,\infty}$, we have established two inequalities involving $\|\Delta^{t,\infty}\|_{2,\infty}$ and  $\|D^{t,\infty}\|_F$ simultaneously. This allows us to relax the requirement on the distance to the ground truth when performing the analysis. In addition, the extension of the proofs in \cite{Ding2020} for positive-definite matrices to general matrices needs to be done properly, including providing bounds for new terms appeared due to asymmetry.
    \item Secondly, while the outlier terms in the initialization can be controlled easily by the sparsity of the outliers, the outlier terms in the induction steps need to be bounded carefully. In particular, our analysis does not require the pattern and the signs of the outliers to be random (as is assumed in \cite{Chen2021}), but is applicable for any sort of outliers.
    \item Last but not least, more effective technical tools have been used in the initialization and induction steps (for example using \cite[Lemma 2]{Chen2015} instead of \cite[Lemma 21]{Ding2020}) and new lemmas have also been developed (for example Lemmas \ref{lem:bound1}~and~\ref{lem:P_Omega_AB}).
\end{itemize}

\subsection{Assumptions and Notations}

Without loss of generality, suppose the ground truth  $L^{\star}\in\mathbb{R}^{n\times n}$ is a square, rank-$r$ matrix.  Denote by $L^{\star}=U^{\star}\Si^{\star}\lb V^{\star}\rb^T$  the compact SVD of $L^{\star}$, where $U^{\star}\in\mathbb{R}^{n\times r},V^{\star}\in\mathbb{R}^{n\times r}$, and $\Si^{\star} = \diag(\sigma_1^{\star},\cdots,\sigma_r^{\star})$. With the convention
$\sigma_1^{\star}\geq\cdots\geq\sigma_r^{\star}$,  the condition number of $L^{\star}$, denoted $\kappa$, is given by $\kappa:=\sigma_1^{\star}/\sigma_r^{\star}$. 

\noindent\textbf{Assumption~1.} Assume $L^{\star}\in\mathbb{R}^{n\times n}$ is $\mu$-incoherent, i.e.,
$$
\ln U^{\star}\rn_{2,\infty}\leq \sqrt{\frac{\mu r}{n}},~
\ln V^{\star}\rn_{2,\infty}\leq \sqrt{\frac{\mu r}{n}},
$$
where for any matrix $M$,  $M_{i,:}$ denotes the $i$-th row of $M$,  and $\|M\|_{2,\infty}:=\max\limits_{i}\ln M_{i,:}\rn_2$.

Let $\|M\|_{\infty}:=\max\limits_{i,j} |M_{ij}|$. It follows immediately from Assumption 1 that 
$$
\ln L \rn_{\infty}\leq \frac{\mu r}{n}\sigma_1^{\star},\quad\mbox{and}\quad\max\left\{\ln L\rn_{2,\infty},\ln L^T\rn_{2,\infty}\right\}\leq \sqrt{\frac{\mu r}{n}}\sigma_1^{\star}.
$$

\noindent\textbf{Assumption~2.} Each entry is independently observed with probability $p$, i.e.,
$$
\mathbb{P}\lcb (i,j)\in \Omega \rcb = p.
$$

\noindent\textbf{Assumption~3.} Denote by $\Omega_{S^{\star}}$ the support of the outliers in $\Omega$. It is assumed that $\Po\lb S^{\star}\rb$ is $2\alpha p$-sparse, i.e., it has no more than $2\alpha p n$ nonzero entries per row and column.

\noindent{\textbf{Remark~1.}} Assumptions 1 and 2 are standard assumptions in the matrix completion literature. Assumption 1, together with Assumption 3, can be considered as separation conditions on $(L^{\star},S^{\star})$ \cite{Candes2011,Chen2021}. One can see that if each observed entry independently has probability $\alpha$ to be corrupted by an outlier, Assumption 3 can be satisfied with high probability\footnote{Throughout the paper, this terminology means that the event occurs with probability at least $1 - C_1n^{-C_2}$ for some constants $C_1, C_2 > 0$ and $C_2$ sufficiently large.}. However, Assumption 3 can account for other sparsity patterns of the outliers as well.

Lastly, we introduce more notations that will be used throughout this paper. We denote by $e_i$ the the $i$-th standard basis vector, by $\bm{1}$  the all-one vector, and by $I$  the identity matrix. Their dimensions will be  clear from the context. We use $\langle\cdot,\cdot\rangle$ to denote the standard inner product between two matrices. For a matrix $M$, we  use $M_{:,j}$ to denote the $j$-th column of $M$. The spectral norm and Frobenius norm of $M$ are denoted by $\ln M\rn_2$ and $\ln M\rn_{\mathrm{F}}$, respectively.

\section{Nonconvex RMC with General Thresholding Functions}

\subsection{Algorithm and Main Result}
Let $\Po:\mathbb{R}^{n\times n}\rightarrow\mathbb{R}^{n\times n}$  be the projection onto the set indexed by $\Omega$, i.e,
$$
\lb\Po(Z)\rb_{ij} = \left\{\begin{array}{cc}
Z_{ij} & (i,j)\in \Omega \\
0      & \text{otherwise}
\end{array}\right.,
$$
and $\P_r$  be the operator that computes the best rank-$r$ approximation of a matrix. The iterates of the  SVP algorithm \cite{Jain2010} for low rank matrix completion (i.e., $S^{\star}_{ij}=0$ in \eqref{eq:RMC-model}) can be written as
$$
L^{t+1}=\P_r\lb L^t-p^{-1}\Po\lb L^t-M\rb\rb,
$$
which is a projected gradient descent algorithm for a rank-constrained least squares problem,
$$
\min_{Z}~\frac12\ln\Po\lb Z-M\rb\rn_{\mathrm{F}}^2\quad \text{s.t.}\quad\rank(Z) = r. 
$$
When the measurements are corrupted by sparse outliers, the R-RMC algorithm proposed in \cite{Cherapanamjeri2017} combines SVP with hard-thresholding for the outliers,
\begin{equation}\label{eq:R-RMC}
\begin{aligned}
S^{t}=\mathcal{HT}_{\eta^{t}}\lb\Po\lb M-L^{t}\rb\rb,~L^{t+1}=\P_r\lb L^t-p^{-1}\Po\lb \lb L^t+S^t\rb-M\rb\rb,
\end{aligned}
\end{equation}
where 
$$
\mathcal{HT}_{\eta}(z) = \begin{cases}
0 & |z|\leq \eta \\
z & \text{otherwise}
\end{cases}
$$
is the hard-thresholding function with a parameter $\eta^{t}>0$ computed from $L^{t}$. 


The discontinuous nature of hard-thresholding function makes it difficult to extend the leave-one-out analysis in~\cite{Ding2020}  to R-RMC. Therefore, we will consider continuous thresholding functions in this paper. {\em Since there are different options, for example soft-thresholding and  SCAD \textup{\cite{Fan2001}}, which one should we consider? The phase transition tests presented in Figure~\ref{fig:PT1} suggest that both of them work desirably. Therefore, instead of focusing on a particular thresholding function, we consider a general class of thresholding functions. }Specifically, letting $\mathcal{T}_{\lambda}$ be the thresholding function with parameter $\lambda>0$, we  assume the following three properties hold: 
\begin{enumerate}[label=\textbf{P.\arabic*}]
    \item \label{P1} $\forall\, |x|\leq \lambda$, $\mathcal{T}_{\lambda}(x)=0$;
    \item \label{P2} there exists a constant $K>0$ such that $|T_{\lambda}(x)-T_{\lambda}(y)|\leq K|x-y|,~\forall\, x,y$;
    \item \label{P3} there exists a constant $B>0$ such that $|T_{\lambda}(x)-x|\leq B\lambda,~\forall\, x$.
\end{enumerate}

As can be seen later, 1) the first property will ensure that the estimated outliers have the correct support; 2) the second property will help us transfer the desirable properties of the low-rank part to the outlier part, and 3) the last property will allow us to control outliers with large magnitudes. Moreover, the aforementioned soft-thresholding and SCAD functions satisfy these properties, and a justification of this claim is provided in Lemma~\ref{lem:general_threh}. 

\begin{figure}[!t]
\subfloat{\includegraphics[width=.5\linewidth]{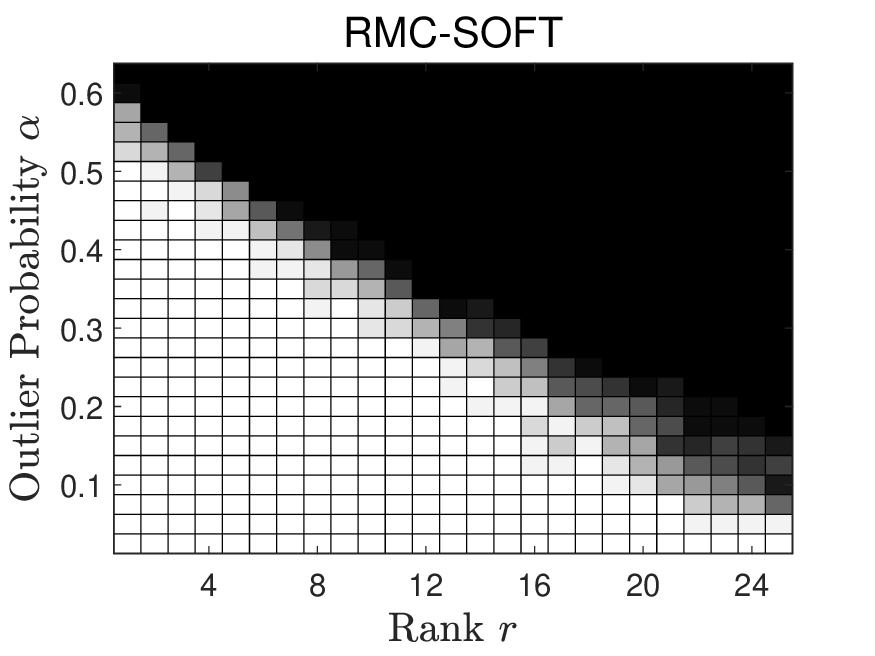}} \hfill
\subfloat{\includegraphics[width=.5\linewidth]{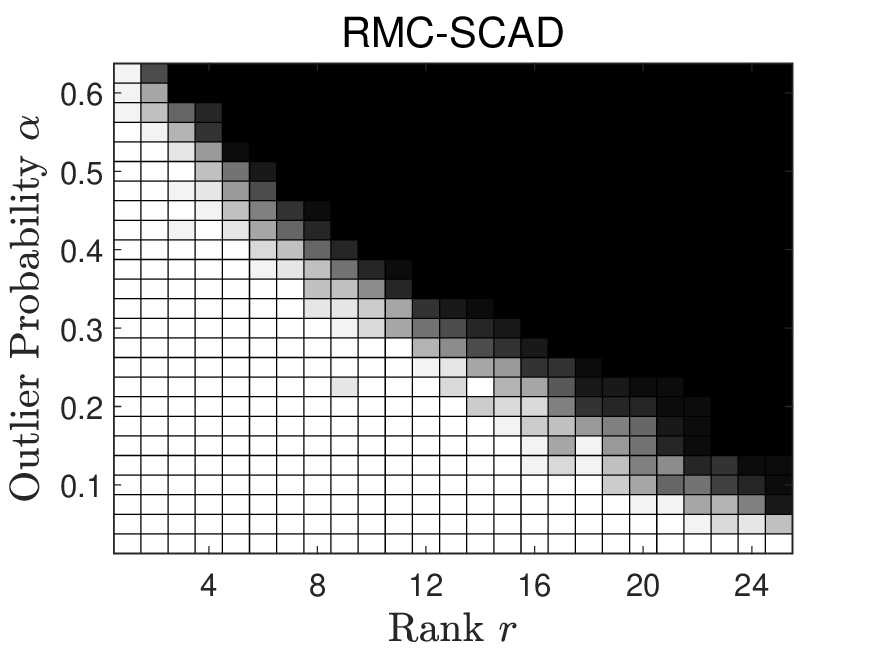}}
\caption{Empirical phase transition comparisons for \textbf{RMC-SOFT} and \textbf{RMC-SCAD}: rank $r$ vs. outlier probability $\alpha$. Tests are conducted for fixed $n=1000$, $p=0.3$, $r\in\{1,\cdots,25\}$, and $\alpha\in\{0.025,0.05,\cdots,0.625\}$. The white color corresponds to success with probability 1, and the black color corresponds to success with probability 0. } \label{fig:PT1}
\end{figure}

The nonconvex method considered in this paper is formally described in Algorithm~\ref{Alg1}.
The theoretical guarantee for it is stated below,  which is applicable to any thresholding function satisfying the properties \labelcref{P1,P2,P3}. 

\begin{theorem}\label{thm:noiseless}
Suppose that Alg.~\ref{Alg1} is performed with a thresholding function satisfying properties \textup{\labelcref{P1,P2,P3}}, and set $\frac{\mu r}{n}\sigma^{\star}_1\leq\beta\leq C_{\emph{init}}\cdot\frac{\mu r}{n}\sigma^{\star}_1$ for some constant $C_{\emph{init}}\geq 1$. Let $C_{\emph{thresh}}:=(K+B)\cdot C_{\emph{init}}$. If Assumptions 1-3 are satisfied with
$$
p\geq\frac{C_{\emph{sample}}}{\gamma^2}\cdot\frac{\kappa^4\mu^3r^3\log n}{n},\quad
\alpha\leq\frac{c_{\emph{outlier}}}{\kappa^2\mu^2r^2}\cdot\frac{\gamma}{C_{\emph{thresh}}}
$$
for some sufficiently large constant $C_{\emph{sample}}>0$, and some sufficiently small constant $c_{\emph{outlier}}>0$, then with high probability, the iterates of Alg.~\ref{Alg1} satisfy
$$
\ln L^{t}-L^{\star} \rn_{\infty}\leq \lb\frac{\mu r}{n}\sigma_{1}^{\star}\rb\gamma^t,
$$
and
$$
\text{Supp}\lb S^{t}\rb\subseteq\Omega_{S^{\star}},~\ln \Po\lb S^{t}-S^{\star}\rb\rn_{\infty}\leq C_{\emph{thresh}}\cdot\lb\frac{\mu r}{n}\sigma_{1}^{\star}\rb\gamma^t
$$
for iteration $0\leq t\leq T$, where $T=n^{O(1)}$. 
\end{theorem}

\begin{algorithm}[!t]
\caption{Nonconvex RMC with General Thresholding Functions} \label{Alg1}
\begin{algorithmic}[1]
\State \textbf{Initialization:} $L^{0}=0$, threshold parameter $\beta>0$, and decay rate $\gamma\in(0,1)$.
\For{$t=0,1,\cdots$}
    \State $\xi^{t} = \beta\cdot\gamma^{t}$
    \State $S^{t} = \T_{\xi^{t}}\lb\Po\lb M-L^{t}\rb\rb$
    \State $L^{t+1}=\P_r\lb L^t-p^{-1}\Po\lb L^t+S^t-M\rb\rb$
\EndFor
\end{algorithmic}
\end{algorithm}


\noindent\textbf{Remark~3.} 
Note that our result is also applicable for the noiseless case, i.e. the low rank matrix completion problem. In this case, we have improved the sample complexity of SVP from $O(\frac{\kappa^6\mu^4r^6\log n}{n})$ in \cite{Ding2020} to $O(\frac{\kappa^4\mu^3r^3\log n}{n})$.

\subsection{Numerical Experiments}\label{sec:numerics}

In this section, phase transition tests are conducted to compare Alg.~1 with the other two nonconvex methods  \textbf{RPCA-GD} \cite{Yi2016} and \textbf{R-RMC} \cite{Cherapanamjeri2017} mentioned in the introduction. We refer to Alg.~1 with soft-thresholding as \textbf{RMC-SOFT}, and Alg.~1 with SCAD shrinkage as \textbf{RMC-SCAD}.  We set $\beta=1.1\frac{\mu r}{n}\sigma_1^{\star}$ and $\gamma=0.9$ in Alg.~1, and the parameter $a$ for SCAD (see \eqref{eq:threshold_funs}) is set to 3. The codes of \textbf{RPCA-GD} are available online and we implement a modified version of \textbf{R-RMC} by removing the rank-increasing scheme since it affects little about the final results in our experiments. Also recall that the theoretical guarantee of \textbf{R-RMC} is based on a sample splitting scheme that is not used in practice. 

\begin{figure}[!t]
\subfloat{\includegraphics[width=.5\linewidth]{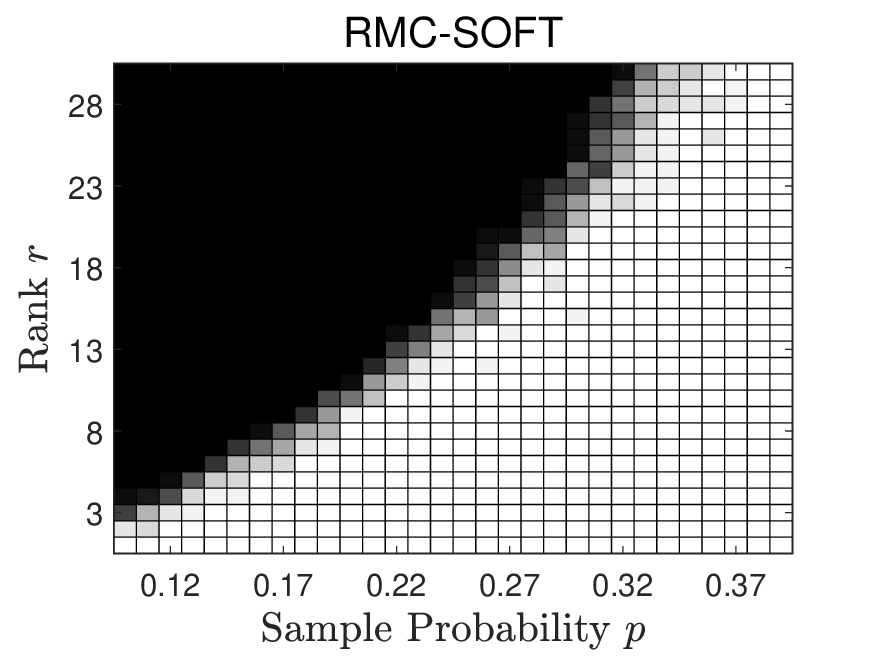}} \hfill
\subfloat{\includegraphics[width=.5\linewidth]{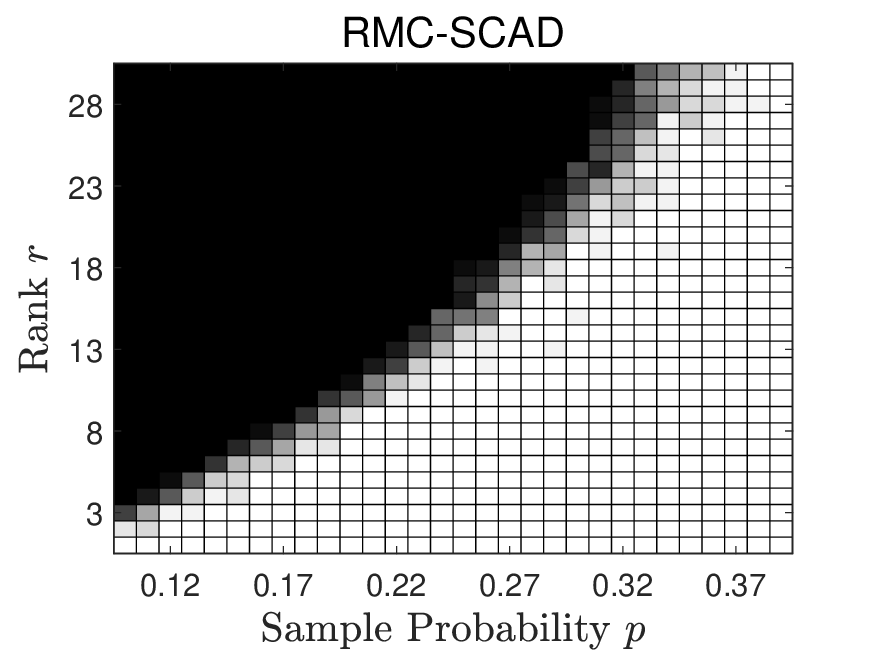}} \hfill
\subfloat{\includegraphics[width=.5\linewidth]{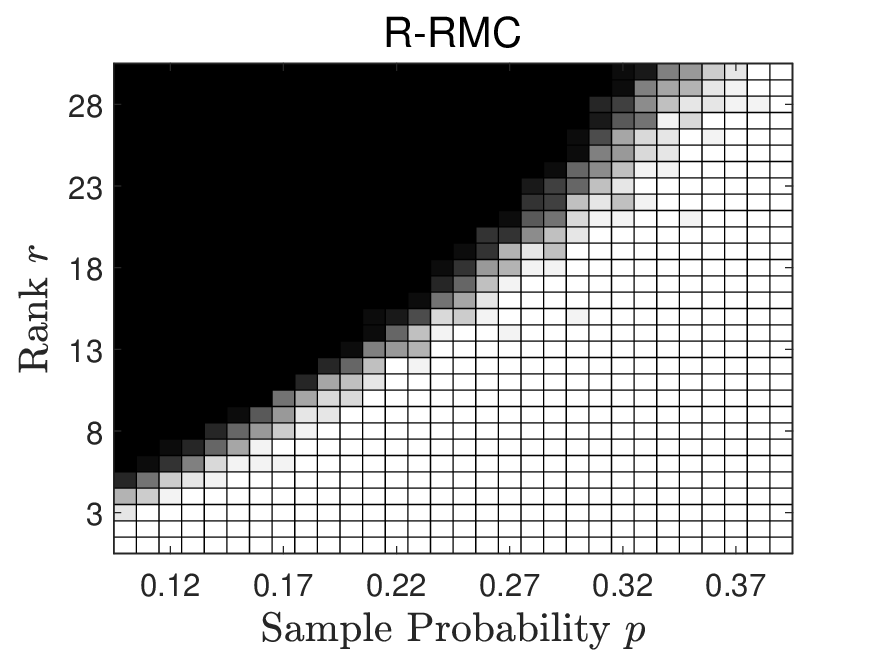}} \hfill
\subfloat{\includegraphics[width=.5\linewidth]{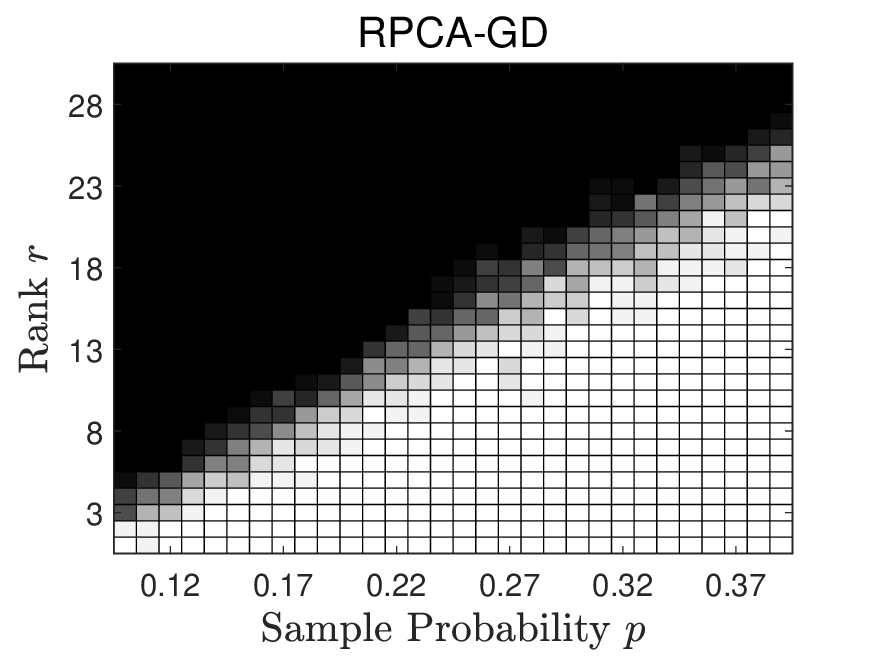}}
\caption{Empirical phase transition comparisons for \textbf{RMC-SOFT}, \textbf{RMC-SCAD}, \textbf{R-RMC}, and \textbf{RPCA-GD}: sample probability $p$ vs. rank $r$.} \label{fig:PT2}
\end{figure}

Tests are conducted for $n=1000,~\alpha=0.1$, $p\in\{0.1,0.11,\cdots,0.39\}$, and $r\in\{1,\cdots,30\}$. For each test instance, we  generate the ground truth matrix $L^\star$ via $L^{\star} = X^{\star}\lb Y^{\star}\rb^T$, where $X^{\star}\in\mathbb{R}^{n\times r}, Y^{\star}\in\mathbb{R}^{n\times r}$ are standard Gaussian matrices. Each location is then included in $\Omega$ independently with probability $p$. Finally, with probability $\alpha$, every $(i,j)\in\Omega$ is corrupted with an  outlier $S^{\star}_{ij}$ chosen uniformly from $\lsb-2\ln L^{\star}\rn_{\infty},2\ln L^{\star}\rn_{\infty}\rsb$. In all the experiments, we run an algorithm till convergence, and record its percentage of successful recoveries among 20 independent trails. Each algorithm is  considered to have successfully recovered the ground truth matrix if at convergence, the estimated low-rank matrix $L_{T}$ satisfies: ${\ln L_{T} - L^{\star}\rn_{\infty}}/{\ln L^{\star}\rn_{\infty}} \leq 10^{-3}.$ 
The phase transition results are presented in Fig.~\ref{fig:PT2}. It can be observed that  \textbf{RMC-SOFT}, \textbf{RMC-SCAD} and \textbf{R-RMC} overall exhibit similar performance, which is superior to \textbf{RPCA-GD}.

\section{Proof Outline}\label{proof_outline}

We first provide an outline of the proof, and defer the details to later sections. As already mentioned, we only consider the case $n_1=n_2=n$ for simplicity. It is straightforward to extend the analysis to the rectangular case. 


First the update of $L^t$ in  Alg.~\ref{Alg1}
can be written in the following  form:
\begin{align*}
L^{t+1} = \P_r\lb L^t-p^{-1}\Po\lb L^t+S^t-M\rb\rb
=\P_r \lb L^{\star}-p^{-1}\Po\lb S^{t}-S^{\star}\rb+\Ho\lb L^{t}-L^{\star}\rb\rb,
\end{align*}
where $\Ho:=\I-\frac1p\Po$. Thus,  the whole  algorithm can be written as: Start from $L^0=0$ and repeat
$$
\left\{\begin{array}{ccl}
S^{t} &= & \T_{\xi^{t}}\lb\Po\lb M-L^{t}\rb\rb, \\
L^{t+1} &= & \P_r\lb L^{\star}-p^{-1}\Po\lb S^{t}-S^{\star}\rb+\Ho\lb L^{t}-L^{\star}\rb\rb
\end{array}\right.~(t\geq 0).
$$
In order to establish its convergence, inspired by \cite{Ma2019,Ding2020}, we introduce the following auxiliary leave-one-out sequences for $i\in\{1,\dots,2n\}$: Starting from $L^{0,i}=0$,
$$
\left\{\begin{array}{ccl}
S^{t,i} &=& \T_{\xi^{t}}\lb\Po^{(-i)}\lb M-L^{t,i}\rb\rb, \\
L^{t+1,i}&=&\P_r\lb L^{\star}-p^{-1}\Po^{(-i)}\lb S^{t,i}-S^{\star}\rb+\Ho^{(-i)}\lb L^{t,i}-L^{\star}\rb\rb
\end{array}\right.~(t\geq 0),
$$
where 
$$
\lb\Po^{(-i)}(Z)\rb_{jk} = \left\{\begin{array}{cc}
\delta_{jk} Z_{jk} & j\neq i \\
0 & j = i
\end{array}\right.,~
\lb\Ho^{(-i)}(Z)\rb_{jk} = \left\{\begin{array}{cc}
(1-\frac1p\delta_{jk})Z_{jk} & j\neq i \\
0 & j = i
\end{array}\right.
$$
if $1\leq i\leq n$, and
$$
\lb\Po^{(-i)}(Z)\rb_{jk} = \left\{\begin{array}{cc}
\delta_{jk} Z_{jk} & k\neq i-n \\
0 & k = i-n
\end{array}\right.,~
\lb\Ho^{(-i)}(Z)\rb_{jk} = \left\{\begin{array}{cc}
(1-\frac1p\delta_{jk})Z_{jk} & k\neq i-n \\
0 & k = i-n
\end{array}\right.
$$
if $n+1\leq i\leq 2n$. 

Note that the original sequence of the proposed algorithm can then be written in a unified way, with the definition $\Po^{(-0)}:=\Po$, and $\Ho^{(-0)}:=\Ho$. We can see that by construction, $\{S^{t,i}\}_{t=0}^{\infty}$ and $\{L^{t,i}\}_{t=0}^{\infty}$ are independent with respect to the random variables $\delta_{ij}$ on the $i$-th row if $1\leq i\leq n$, and the $(i-n)$-th column if $n+1\leq i\leq 2n$. Such a decoupling of the statistical dependence plays a key role in proving Theorem~\ref{thm:noiseless} in an inductive way. More notations should be introduced in order to present the induction hypotheses.

For $i\in\{0,1,\cdots,2n\}$, letting $E_1^{t,i}:=p^{-1}\Po^{(-i)}\lb S^{t,i}-S^{\star}\rb$, and $E_2^{t,i}:=\Ho^{(-i)}\lb L^{t,i}-L^{\star}\rb$, the update for the low-rank part  can be written as:
$$
L^{t+1,i}=\P_r\lb L^{\star}-E_1^{t,i}+E_2^{t,i}\rb:=\P_r\lb L^{\star}+E^{t,i}\rb,
$$
which can be deemed as a perturbation of $L^{\star}$. Denote by  $U^{t+1,i}\Si^{t+1,i}\lb V^{t+1,i}\rb^T$ the SVD of $L^{t+1,i}$. The deviation between 
$$
F^{t+1,i} := \left[\begin{array}{c}
U^{t+1,i} \\
V^{t+1,i}
\end{array}\right]~\text{and}~
F^{\star} := \left[\begin{array}{c}
U^{\star} \\
V^{\star}
\end{array}\right]
$$
can be measured by
$$
\min_{R\in\mathcal{O}(r)}\ln F^{t+1,i}-F^{\star}R\rn_{\mathrm{F}},
$$
where $\mathcal{O}(r):=\{R\in\mathbb{R}^{r\times r}~|~R^TR=I_r\}$. This is known as the orthogonal Procrustes problem \cite{Ma2019}, and one optimal rotation matrix, denoted $G^{t+1,i}$, can be obtained by computing the SVD of $H^{t+1,i}:=\frac12\lb F^{\star}\rb^TF^{t+1,i}=A^{t+1,i}\widetilde{\Si}^{t+1,i}\lb B^{t+1,i}\rb^T$, and setting $G^{t+1,i}:=A^{t+1,i}\lb B^{t+1,i}\rb^T$. We further define
$\De^{t+1,i} := F^{t+1,i}-F^{\star}G^{t+1,i}$. 

When $\ln E^{t,i} \rn_2$ is small enough, $\|L^{t+1,i}-L^{\star}\|_{\infty}$ can be bounded by $\ln\De^{t+1,i}\rn_{2,\infty}$, as stated in the following lemma. This is an extension of \cite[Lemma 19]{Ding2020} to general matrices. For completeness, we include its proof in Appendix~\ref{sec:proof_L_infty}.

\begin{lemma}\label{lem:L_infinity}
Let $L^{\star} = U^{\star}\Sigma^{\star}(V^{\star})^T$ be the SVD, where $U^{\star}\in\mathbb{R}^{n\times r}, \Sigma^{\star}\in \mathbb{R}^{r \times r}, V^{\star}\in\mathbb{R}^{n\times r}$. Consider $L=\mathcal{P}_r\left(L^{\star}+E\right)$ for some perturbation matrix $E$. Let $L = U\Sigma V^T$ be its SVD. Define $F^{\star}=[(U^{\star})^T~(V^{\star})^T]^T$, and $F=[U^T~V^T]^T$. Suppose the rank-$r$ SVD of $\left(F^{\star}\right)^TF$ is $A\widetilde{\Si} B^T$. Set $G=AB^T$, and $\Delta=F-F^{\star}G$. If $\|E\|_{2} \leq\frac12\sigma_r^{\star}$, then
$$
\begin{aligned}
\left\|L-L^{\star}\right\|_{\infty} \leq & \|\Delta\|_{2, \infty}\lb\|F\|_{2, \infty}+\|F^{\star}\|_{2, \infty}\rb\|\Sigma\|_{2}+(3+4\kappa)\left\|F^{\star}\right\|_{2, \infty}^2\|E\|_{2}.
\end{aligned}
$$
\end{lemma}

Once we have the bound for $\|L^{t+1,i}-L^{\star}\|_{\infty}$, the entrywise error for the outlier estimation can also be controlled, as stated in the following lemma. The proof of this lemma is deferred to Appendix~\ref{sec:proof_thresh}.

\begin{lemma}\label{lem:thresh}
Suppose
$$
\frac{\mu r}{n}\sigma_1^{\star}\leq\beta\leq C_{\emph{init}}\cdot\frac{\mu r}{n}\sigma_1^{\star},
$$
$\xi^{t} = \beta\cdot\gamma^{t}$, and the thresholding function satisfies the properties \textup{\labelcref{P1,P2,P3}}. For $i\in\{0,1,\cdots,2n\}$, if
$$
\ln L^{t,i}-L^{\star} \rn_{\infty}\leq \lb\frac{\mu r}{n}\sigma_{1}^{\star}\rb\gamma^t,
$$
then
$$
\text{Supp}\lb S^{t,i}\rb\subseteq\Omega^{(-i)}\cap\Omega_{S^{\star}},~\ln\Po^{(-i)}\lb S^{t,i}-S^{\star}\rb\rn_{\infty}\leq C_{\emph{thresh}}\cdot\lb\frac{\mu r}{n}\sigma_{1}^{\star}\rb\gamma^t.
$$
Here, $\Omega^{(-i)}=\Omega$ if $i=0$, $\Omega^{(-i)}$ is $\Omega$ without the indices from the $i$-th row if $1\leq i\leq n$,  $\Omega^{(-i)}$ is $\Omega$ without the indices from the $(i-n)$-th column if $n+1\leq i\leq 2n$, and $C_{\emph{thresh}}:=(K+B)\cdot C_{\emph{init}}$. 
\end{lemma}

Bounding $\ln\De^{t+1,i}\rn_{2,\infty}$ is the crux of our proofs, and we achieve this by establishing proximity between different sequences. For $i,m\in\{0,1,\cdots,2n\}$, the deviation between $F^{t+1,i}$ and $F^{t+1,m}$ can be measured by
$$
\min_{R\in\mathcal{O}(r)}\ln F^{t+1,i}-F^{t+1,m}R\rn_{\mathrm{F}}.
$$
Let $G^{t+1,i,m}$ be the optimal rotation matrix, and define $D^{t+1,i,m} := F^{t+1,i}-F^{t+1,m}G^{t+1,i,m}$. We will show that the following induction hypotheses hold for all the sequences.

\begin{theorem}\label{thm:induction}
Under the assumptions of Theorem~\ref{thm:noiseless}, for iteration $1\leq t\leq T$, where $T=n^{O(1)}$, 
\begin{subequations}
\begin{align}
    \max_{0\leq i\leq 2n}\ln E^{t-1,i} \rn_2:=\ln E^{t-1,\infty}\rn_2 &\leq \frac{1}{C_0}\cdot\lb\frac{1}{\kappa\mu r}\sigma_{r}^{\star}\rb\gamma^t\label{eq:op}\\
    \max_{0\leq i\leq 2n}\ln\De^{t,i}\rn_{2,\infty}:=\ln \De^{t,\infty}\rn_{2,\infty}&\leq\frac{5C_{10}}{C_0}\cdot\lb\frac{1}{\mu r}\sqrt{\frac{\mu r}{n}}\rb\gamma^t\label{eq:l_2_infty}\\
    \max_{0\leq i,m\leq 2n}\ln D^{t,i,m}\rn_{\mathrm{F}}:=\ln D^{t,\infty}\rn_{\mathrm{F}}&\leq\frac{8C_{10}}{C_0}\cdot\lb\frac{1}{\kappa \mu r}\sqrt{\frac{\mu r}{n}}\rb\gamma^t\label{eq:proximity}
\end{align}
\end{subequations}
hold true with high probability, for some positive constants $C_{10}$ and $C_0$ satisfying \eqref{eq:constants}.
\end{theorem}

Once Theorem~\ref{thm:induction} is established, the proof of Theorem~\ref{thm:noiseless} can be proceeded as follows.

\begin{proof}[{\normalfont\textbf{Proof of Theorem~\ref{thm:noiseless}}}]
First consider the case $t=0$. Since 
$$
\ln L^{0}-L^{\star} \rn_{\infty} = \ln L^{\star} \rn_{\infty}\leq \frac{\mu r}{n} \sigma_1^{\star},
$$
by Lemma~\ref{lem:thresh},
$$
\text{Supp}\lb S^{0}\rb\subseteq\Omega_{S^{\star}},~\ln \Po\lb S^{0}-S^{\star}\rb\rn_{\infty}\leq C_{\text{thresh}}\cdot\frac{\mu r}{n}\sigma_{1}^{\star}.
$$
For the case $t\geq 1$, due to~\eqref{eq:op}, Lemma~\ref{lem:L_infinity} is applicable. Thus we have
\begin{align*}
\ln L^{t}-L^{\star} \rn_{\infty} \leq & \|\Delta^{t}\|_{2, \infty}\lb\|F^t\|_{2, \infty}+\|F^{\star}\|_{2, \infty}\rb\|\Sigma^{t}\|_{2}+(3+4\kappa)\left\|F^{\star}\right\|_{2, \infty}^2\|E^{t-1}\|_{2} \\
\leq & \frac{5C_{10}}{C_0}\lb\frac{1}{\mu r}\sqrt{\frac{\mu r}{n}}\rb\gamma^t\cdot\lb 2+\frac{5C_{10}}{C_0}\rb\sqrt{\frac{\mu r}{n}}\cdot\lb 1+\frac{1}{C_0}\rb\sigma_1^{\star}+\frac{\mu r}{n}\cdot\frac{7}{C_0}\lb\frac{1}{\mu r}\sigma_{r}^{\star}\rb\gamma^t\\
\leq & \frac{12C_{10}}{C_0}\cdot\frac{1}{\mu r}\lb\frac{\mu r}{n}\sigma_1^{\star}\rb\gamma^t\leq\lb\frac{\mu r}{n}\sigma_1^{\star}\rb\gamma^t,
\end{align*}
where the bounds \eqref{eq:op} and \eqref{eq:l_2_infty} are used in the second inequality, and the bounds of $C_{10}$ and $C_0$ (see \eqref{eq:constants}) are used in the last two inequalities. By Lemma~\ref{lem:thresh} again,
\begin{equation*}
\text{Supp}\lb S^{t}\rb\subseteq\Omega_{S^{\star}},~\ln \Po\lb S^{t}-S^{\star}\rb\rn_{\infty}\leq C_{\text{thresh}}\cdot\lb\frac{\mu r}{n}\sigma_{1}^{\star}\rb\gamma^{t},
\end{equation*}
and the proof is finished.\qedhere
\end{proof}

In the following two sections, we verify the induction hypotheses for $t=1$ and $t> 1$, respectively. Note that the latter case needs to be proved differently due to the more complicated statistical dependency in the iterations. The key differences between our proofs and the ones in \cite{Ding2020} are also discussed at the beginning of each section.

\section{Initialization and Base Case}\label{sec:base}

{\em In the base case, our proofs and the ones in \textup{\cite{Ding2020}} are mainly different in the following  two aspects: \begin{itemize} \item We use Lemma \ref{lem:init} instead of \textup{\cite[Lemma 21]{Ding2020}} to bound $\ln\Ho\lb L^{\star}\rb\rn_2$, and use Bernstein's inequality applied to vectors instead of scalars to bound $\ln e_m^T\Ho^{(-i)}\lb L^{\star}\rb V^{1,m} \rn_2$, yielding tighter bounds in terms of $\mu$ and $r$.
\item We need to bound the extra term $T_3$ appearing due to the asymmetry of $L^{\star}$, as well as those terms involving outliers.
\end{itemize}}
\noindent Next let us establish \cref{eq:op,eq:l_2_infty,eq:proximity} for $t=1$.
For $\forall i\in\{0,1,\cdots,2n\}$,
$$
\ln L^{0,i}-L^{\star} \rn_{\infty} = \ln L^{\star} \rn_{\infty}\leq \frac{\mu r}{n} \sigma_1^{\star}.
$$
It follows from Lemma~\ref{lem:thresh} that
$$
\text{Supp}\lb S^{0,i}\rb\subseteq\Omega^{(-i)}\cap\Omega_{S^{\star}},~\ln\Po^{(-i)}\lb S^{0,i}-S^{\star}\rb\rn_{\infty}\leq C_{\text{thresh}}\cdot\frac{\mu r}{n}\sigma_{1}^{\star}.
$$

\noindent\textbf{Step 1: The Bound For $\ln E^{0,\infty} \rn_2$.} First consider $\ln E_1^{0,i} \rn_2$. Recall from Assumption 3 that $\Po\lb S^{\star}\rb$ is $2\alpha p$-sparse. One has
\begin{equation}\label{eq:S_init_op}
\begin{aligned} 
\ln E_1^{0,i} \rn_2 = & \ln p^{-1}\Po^{(-i)}\lb S^{0,i}-S^{\star}\rb\rn_2\leq p^{-1}\cdot\lb2\alpha p n\rb \cdot C_{\text{thresh}}\frac{\mu r}{n}\sigma_{1}^{\star}\leq \frac{\gamma}{2C_0}\cdot\frac{1}{\kappa\mu r}\sigma_{r}^{\star},
\end{aligned}
\end{equation}
where the first inequality follows from applying Lemma~\ref{lem:S_op} to the $2\alpha p$-sparse matrix $\Po^{(-i)}\lb S^{0,i}-S^{\star}\rb$, and the second inequality holds if 
$$
\alpha\leq\frac{1}{4C_0}\frac{1}{\kappa^2\mu^2r^2}\cdot\frac{\gamma}{C_{\text{thresh}}}.
$$
For $\ln E_2^{0,i} \rn_2$, one has
\begin{equation}\label{eq:op_init}
\begin{aligned}
\ln E_2^{0,i} \rn_2 = \ln\Ho^{(-i)}\lb-L^{\star}\rb\rn_2 \leq & \ln\Ho\lb L^{\star}\rb\rn_2\\
\leq & c_4\lb \frac{\log n}{p}\ln L^{\star}\rn_{\infty}+\sqrt{\frac{\log n}{p}}\cdot\max\left\{\ln L^{\star}\rn_{2,\infty},\ln \lb L^{\star}\rb^T\rn_{2,\infty}\right\}\rb\\
\leq & c_4\lb \frac{\log n}{p}\frac{\mu r}{n}\sigma_1^{\star}+\sqrt{\frac{\log n}{p}}\sqrt{\frac{\mu r}{n}}\sigma_1^{\star}\rb\\
\leq & 2c_4\sqrt{\frac{\mu r\log n}{np}}\sigma_1^{\star}\leq \frac{\gamma}{2C_0}\cdot\frac{1}{\kappa\mu r}\sigma_{r}^{\star},
\end{aligned}
\end{equation}
where in the second inequality we use Lemma \ref{lem:init}  which holds with high probability for some universal constant $c_4>1$, and the last inequality holds if 
$$
p\geq\frac{16c_4^2C_0^2}{\gamma^2}\cdot\frac{\kappa^4\mu^3r^3\log n}{n}.
$$
Therefore,
\begin{equation}\label{eq:base_op}
\ln E^{0,i} \rn_2 \leq \ln E_1^{0,i} \rn_2+\ln E_2^{0,i} \rn_2\leq\frac{1}{C_0}\cdot\lb\frac{1}{\kappa\mu r}\sigma_{r}^{\star}\rb\gamma.
\end{equation}

\noindent\textbf{Step 2: The Bound For $\ln \De^{1,\infty} \rn_{2,\infty}$.} To this end, we need to consider the symmetrization of a general matrix. For a matrix $ M\in\mathbb{R}^{n\times n}$, its symmetrization is defined as 
$$
\widehat{M} := \left[\begin{array}{cc}
0   & M \\
M^T & 0
\end{array}\right].
$$
Based on this notion, it is easy to see that the eigen-decomposition of $\widehat{L^{\star}}$ is
$$
\widehat{L^{\star}}=\lb\frac{1}{\sqrt{2}}F^{\star}\rb\Si^{\star}\lb\frac{1}{\sqrt{2}}F^{\star}\rb^T+\lb\frac{1}{\sqrt{2}}\widetilde{F}^{\star}\rb\lb-\Si^{\star}\rb\lb\frac{1}{\sqrt{2}}\widetilde{F}^{\star}\rb^T,
$$
where $\widetilde{F}^{\star}:=[(-U^{\star})^T~(V^{\star})^T]^T$. The top-$r$ eigen-decomposition of $\widehat{L^{\star}}+\widehat{E^{0,i}}$ is $\frac12F^{1,i}\Si^{1,i}\lb F^{1,i}\rb^T$, and $\Si^{1,i}$ is invertible since $\ln\widehat{E^{0,i}}\rn_2=\ln E^{0,i}\rn_2$ is bounded by \eqref{eq:base_op}. Thus for $1\leq m\leq 2n$,
\begin{align*}
\De^{1,i}_{m,:} = &  e_m^T\lb F^{1,i}-F^{\star}G^{1,i}\rb \\
= & e_m^T\lb\lb\widehat{L^{\star}}+\widehat{E^{0,i}}\rb F^{1,i}\lb\Si^{1,i}\rb^{-1}-F^{\star}G^{1,i}\rb \\
= & e_m^TF^{\star}\Si^{\star}\lsb\frac12\lb F^{\star}\rb^T F^{1,i}\lb\Si^{1,i}\rb^{-1}-\lb\Si^{\star}\rb^{-1}G^{1,i}\rsb\\
& +\frac12e_m^T\widetilde{F}^{\star}\lb-\Si^{\star}\rb\lb \widetilde{F}^{\star}\rb^T F^{1,i}\lb\Si^{1,i}\rb^{-1} +e_m^T\widehat{E^{0,i}} F^{1,i}\lb\Si^{1,i}\rb^{-1} \\
= &\underbrace{e_m^TF^{\star}\Si^{\star}\lsb H^{1,i}\lb\Si^{\star}\rb^{-1}-\lb\Si^{\star}\rb^{-1}G^{1,i}\rsb}_{T_1}+\underbrace{e_m^TF^{\star}\Si^{\star}H^{1,i}\lsb\lb\Si^{1,i}\rb^{-1}-\lb\Si^{\star}\rb^{-1}\rsb}_{T_2}\\
&-\frac12\underbrace{e_m^T\widetilde{F}^{\star}\Si^{\star}\lb \widetilde{F}^{\star}\rb^T F^{1,i}\lb\Si^{1,i}\rb^{-1}}_{T_3}+\underbrace{e_m^T\widehat{E^{0,i}} F^{1,i}\lb\Si^{1,i}\rb^{-1}}_{T_4}.
\end{align*}

\noindent\textbf{$\bullet$ Bounding $T_1$.} Consider
\begin{align*}
R := & \Si^{\star}\lsb H^{1,i}\lb\Si^{\star}\rb^{-1}-\lb\Si^{\star}\rb^{-1}G^{1,i}\rsb = \lsb\Si^{\star}H^{1,i}-G^{1,i}\Si^{\star}\rsb\lb\Si^{\star}\rb^{-1}.
\end{align*}
Applying Lemma \ref{lem:perturb_gt} to $\widehat{L^{1,i}}$ and $\widehat{L^{\star}}$ yields
\begin{align*}
    \ln R \rn_2\leq\ln \Si^{\star}H^{1,i}-G^{1,i}\Si^{\star}\rn_2\cdot\ln\lb\Si^{\star}\rb^{-1}\rn_2\leq \frac{4}{C_0}\cdot\lb\frac{1}{\mu r}\rb\gamma.
\end{align*}
We then get
\begin{align*}
    \ln T_1 \rn_2\leq \frac{4}{C_0}\cdot\lb\frac{1}{\mu r}\sqrt{\frac{\mu r}{n}}\rb\gamma.
\end{align*}

\noindent\textbf{$\bullet$ Bounding $T_2$.} Note that
\begin{align*}
    \ln\lb\Si^{1,i}\rb^{-1}-\lb\Si^{\star}\rb^{-1}\rn_2=\max_{1\leq k\leq r} \left|\frac{1}{\sigma_k^{1,i}}-\frac{1}{\sigma_k^{\star}}\right|= \max_{1\leq k\leq r} \frac{\left|\sigma_k^{1,i}-\sigma_k^{\star}\right|}{\sigma_k^{1,i}\sigma_k^{\star}} \leq\frac{2}{C_0}\cdot\lb\frac{1}{\kappa\mu r}\rb\gamma\cdot\frac{1}{\sigma_r^{\star}}.
\end{align*}
Therefore,
\begin{align*}
    \ln T_2 \rn_2 \leq \frac{2}{C_0}\cdot\lb\frac{1}{\mu r}\sqrt{\frac{\mu r}{n}}\rb\gamma.
\end{align*}

\noindent\textbf{$\bullet$ Bounding $T_3$.} Note that
\begin{align*}
    \ln \lb\widetilde{F}^{\star}\rb^T F^{1,i} \rn_2 
    = & \ln \lb U^{\star}\rb^T U^{1,i}-\lb V^{\star}\rb^T V^{1,i} \rn_2 \\
     = & \ln \lb U^{\star}\rb^T U^{1,i}\lb G^{1,i}\rb^T-\lb V^{\star}\rb^T V^{1,i}\lb G^{1,i} \rb^T\rn_2 \\
    = & \ln \lb U^{\star}\rb^T \lsb U^{\star}+\lb U^{1,i}\lb G^{1,i}\rb^T-U^{\star}\rb\rsb-\lb V^{\star}\rb^T \lsb V^{\star}+\lb V^{1,i}\lb G^{1,i}\rb^T-V^{\star}\rb\rsb\rn_2 \\
    \leq &\ln\lb U^{\star}\rb^T \lb U^{1,i}-U^{\star}G^{1,i}\rb \rn_2+\ln\lb V^{\star}\rb^T \lb V^{1,i}-V^{\star}G^{1,i}\rb \rn_2\\
    \leq & 2\ln F^{1,i}-F^{\star}G^{1,i}\rn_2,
\end{align*}
where the last inequality holds since $\ln U^{1,i}-U^{\star}G^{1,i}\rn_2,\ln V^{1,i}-V^{\star}G^{1,i}\rn_2\leq\ln F^{1,i}-F^{\star}G^{1,i}\rn_2$. Then by Lemma \ref{lem:op},
\begin{align*}
\ln F^{1,i}-F^{\star}G^{1,i}\rn_2 \leq \frac{4}{\sigma_r^{\star}} \ln \widehat{L^{1,i}}-\widehat{L^{\star}}\rn_2.
\end{align*}
Moreover,
\begin{align*}
\ln \widehat{L^{1,i}}-\widehat{L^{\star}}\rn_2 
\leq & \ln \widehat{L^{1,i}} -\lb\widehat{L^{\star}}+\widehat{E^{0,i}}\rb\rn_2+\ln \lb\widehat{L^{\star}}+\widehat{E^{0,i}}\rb-\widehat{L^{\star}}\rn_2 \\
\leq &2\ln \lb\widehat{L^{\star}}+\widehat{E^{0,i}}\rb-\widehat{L^{\star}}\rn_2 = 2\ln E^{0,i}\rn_2, 
\end{align*}
where the second inequality follows from the fact that $\widehat{L^{1,i}}$ is the best rank-$2r$ approximation to $\widehat{L^{\star}}+\widehat{E^{0,i}}$. Therefore,
\begin{align*}
    \ln T_3 \rn_2 \leq \frac{20}{C_0}\cdot\lb\frac{1}{\mu r}\sqrt{\frac{\mu r}{n}}\rb\gamma.
\end{align*}

\noindent\textbf{$\bullet$ Bounding $T_4$.} The main task is to bound $\ln e_m^T\widehat{E^{0,i}}F^{1,i}\rn_2$. Due to the symmetrization scheme,
\begin{align*}
    \ln e_m^T\widehat{E^{0,i}}F^{1,i}\rn_2 = \ln e_m^TE^{0,i}V^{1,i}\rn_2 
\end{align*}
if $1\leq m\leq n$, and
\begin{align*}
    \ln e_m^T\widehat{E^{0,i}}F^{1,i}\rn_2 = \ln e_{(m-n)}^T\lb E^{0,i}\rb^T U^{1,i}\rn_2 
\end{align*}
if $n+1\leq m\leq 2n$. Due to our definition of the leave-one-out sequences, in both cases, 
\begin{align*}
    \ln e_m^T\widehat{E^{0,i}}F^{1,i}\rn_2 = 0
\end{align*}
when $i=m$. Thus, it suffices to consider the case $i\neq m$. Without loss of generality, we only consider the case  $0\leq i\leq n,~1\leq m\leq n$ and the proof can be done similarly for the other cases. For $\ln e_m^TE_1^{0,i}V^{1,i}\rn_2$, one has
\begin{align*}
    \ln e_m^TE_1^{0,i}V^{1,i}\rn_2 \leq &\ln e_m^TE_1^{0,i}V^{\star}G^{1,i}\rn_2+\ln e_m^TE_1^{0,i}\lb V^{1,i}-V^{\star}G^{1,i}\rb\rn_2 \\
    = & p^{-1}\lb\ln e_m^T\Po^{(-i)}\lb S^{0,i}-S^{\star}\rb V^{\star}\rn_2+\ln e_m^T\Po^{(-i)}\lb S^{0,i}-S^{\star}\rb \De_V^{1,i}\rn_2\rb\\
    \leq & p^{-1}\cdot(2\alpha pn)\cdot C_{\text{thresh}}\frac{\mu r}{n}\sigma_{1}^{\star}\cdot\lb\sqrt{\frac{\mu r}{n}}+\ln\De^{1,\infty}\rn_{2,\infty}\rb \\
    \leq & \frac{1}{2C_0}\frac{\gamma}{\kappa\mu r}\sigma_{r}^{\star}\lb\sqrt{\frac{\mu r}{n}}+\ln\De^{1,\infty}\rn_{2,\infty}\rb,
\end{align*}
where $\De_V^{1,i}:=V^{1,i}-V^{\star}G^{1,i}$, and the last inequality follows from the bound in \eqref{eq:S_init_op}. 

For $\ln e_m^TE_2^{0,i}V^{1,i}\rn_2$, one has
\begin{align*}
    e_m^TE_2^{0,i} V^{1,i} = e_m^TE_2^{0,i} V^{1,m}G^{1,i,m} +  e_m^TE_2^{0,i} \lb V^{1,i}-V^{1,m}G^{1,i,m} \rb.
\end{align*}
Such splitting allows us to utilize the property that $V^{1,m}$ is independent of the random variables on the $m$-th row. 
For $k=1,\cdots,n$, define 
\[
v_k = \lb 1-\delta_{mk}/p\rb L_{mk}^{\star} V_{k,:}^{1,m}.
\]
One has
\begin{align*}
\ln v_k\rn_2 \leq &\frac1p\ln L^{\star} \rn_{\infty} \ln V^{1,m} \rn_{2,\infty},\\
\left| \sum_{k=1}^n\mathbb{E} \lsb\ln v_k\rn_2^2\rsb \right| \leq & \sum_{k=1}^n \frac1p \lb L_{mk}^{\star}\rb^2\ln V_{k,:}^{1,m} \rn_2^2 \leq \frac1p\ln L^{\star} \rn_{2,\infty}^2\ln V^{1,m} \rn_{2,\infty}^2.
\end{align*}
By Lemma \ref{lem:bernstein},
\begin{equation}\label{eq:vector_bern}
\begin{aligned}
\ln e_m^TE_2^{0,i} V^{1,m}G^{1,i,m} \rn_2 = & \ln e_m^T\Ho^{(-i)}\lb L^{\star}\rb V^{1,m} \rn_2 = \ln \sum_{k=1}^n v_k\rn_2\\
    \leq & C_{10}\cdot\lb\sqrt{\frac{\log n}{p}}\ln L^{\star} \rn_{2,\infty}\ln V^{1,m} \rn_{2,\infty}+\frac{\log n}{p}\ln L^{\star} \rn_{\infty} \ln V^{1,m} \rn_{2,\infty}\rb \\
    \leq & C_{10}\cdot\lb\sqrt{\frac{\mu r\log n}{np}}\sigma_1^{\star}+\frac{\mu r\log n}{np}\sigma_1^{\star}\rb\lb \sqrt{\frac{\mu r}{n}}+\ln \De^{1,\infty} \rn_{2,\infty}\rb  \\
    \leq & \frac{C_{10}}{2c_4C_0}\frac{\gamma}{\kappa\mu r}\sigma_{r}^{\star}\lb\sqrt{\frac{\mu r}{n}} + \ln \De^{1,\infty} \rn_{2,\infty}\rb,
\end{aligned}
\end{equation}
where the last inequality follows from the bound in \eqref{eq:op_init}. Together with the bound
\begin{align*}
    \ln e_m^TE_2^{0,i} \lb V^{1,i}-V^{1,m}G^{1,i,m} \rb\rn_2 \leq \ln E_2^{0,i}\rn_2\ln D^{1,i,m} \rn_{\mathrm{F}}\leq \frac{\gamma}{2C_0}\sigma_r^{\star}\ln D^{1,i,m} \rn_{\mathrm{F}},
\end{align*}
we can get
\begin{align*}
    \ln T_4 \rn_2
    \leq \frac{\gamma}{C_0}\lsb\lb 1+\frac{C_{10}}{c_4}\rb\lb\frac{1}{\kappa\mu r}\sqrt{\frac{\mu r}{n}}+\ln\De^{1,\infty}\rn_{2,\infty}\rb+\ln D^{1,\infty} \rn_{\mathrm{F}}\rsb.
\end{align*}

Combining the bounds of $T_1$ to $T_4$ yields
\begin{equation}\label{eq:2_infty}
\begin{aligned}
    \ln \De^{1,\infty} \rn_{2,\infty}
    \leq \frac{17+\frac{C_{10}}{c_4}}{C_0}\lb\frac{1}{\mu r}\sqrt{\frac{\mu r}{n}}\rb\gamma
    +\frac{1+\frac{C_{10}}{c_4}}{C_0}\ln\De^{1,\infty}\rn_{2,\infty}+\frac{1}{C_0}\ln D^{1,\infty} \rn_{\mathrm{F}}.
\end{aligned} 
\end{equation}

\noindent\textbf{Step 3: The Bound For $\ln D^{1,\infty} \rn_{\mathrm{F}}$.} For $i\in\{0,1,\cdots,2n\},~m\in\{1,\cdots,2n\}$,
\begin{align*}
\left\|D^{1, i, m}\right\|_{\mathrm{F}} \leq & \left\|F^{1, i}-F^{1, m} G^{1,0,m}G^{1,i,0}\right\|_{\mathrm{F}} \\
\leq & \left\|F^{1, i}-F^{1,0} G^{1,i,0}\right\|_{\mathrm{F}}+\left\|\left(F^{1,0} - F^{1, m}G^{1,0,m}\right) G^{1,i,0}\right\|_{\mathrm{F}} = \ln D^{1,i,0} \rn_{\mathrm{F}}+\ln D^{1,0,m} \rn_{\mathrm{F}}.
\end{align*}
We only need to bound $\ln D^{1,0,m} \rn_{\mathrm{F}}$, since $\ln D^{1,m,0} \rn_{\mathrm{F}}=\ln D^{1,0,m} \rn_{\mathrm{F}}$. With a slight abuse of notation, let $A:=\widehat{L^{\star}}+\widehat{E^{0,m}}$ and $\widetilde{A}:=\widehat{L^{\star}}+\widehat{E^{0}}$. According to Weyl's inequality,
$$
|\lambda_r(A)-\sigma_r^{\star}|\leq\ln E^{0,m}\rn_2,~ 
|\lambda_{r+1}(A)|\leq\ln E^{0,m}\rn_2.
$$
Define $W^{0,m}:=\widehat{E^{0}}-\widehat{E^{0,m}}$. Then
$$
\delta:=\lambda_r(A)-\lambda_{r+1}(A)\geq \sigma_r^{\star}-2\ln E^{0,m}\rn_2 > \ln W^{0,m}\rn_2,
$$
where the last inequality is due to the bound of $\|E^{0,\infty}\|_2$. Applying Lemma \ref{lem:perturb_F}, we can get
\begin{equation}\label{eq:WF_init}
\begin{aligned}
\ln D^{1,0,m}\rn_{\mathrm{F}} \leq \frac{\sqrt{2}\ln W^{0,m} F^{1,m}\rn_{\mathrm{F}}}{\delta-\ln W^{0,m}\rn_2}\leq &\frac{2}{\sigma_r^{\star}}\ln W^{0,m} F^{1,m}\rn_{\mathrm{F}} \\
\leq &\frac{2}{\sigma_r^{\star}}\lb \ln\lb E^{0}-E^{0,m}\rb V^{1,m}\rn_{\mathrm{F}}+\ln\lb E^{0}-E^{0,m}\rb^T U^{1,m}\rn_{\mathrm{F}}\rb.
\end{aligned}
\end{equation}
Next we only derive the bound for $\ln\lb E^{0}-E^{0,m}\rb V^{1,m}\rn_{\mathrm{F}}$, and the bound for $\ln\lb E^{0}-E^{0,m}\rb^T U^{1,m}\rn_{\mathrm{F}}$ can be obtained similarly. Note that
\begin{align*}
&\ln\lb E^{0}-E^{0,m}\rb V^{1,m}\rn_{\mathrm{F}} \\
= &\ln \lsb p^{-1}\Po^{(-m)}\lb S^{0,m}-S^{\star} \rb-p^{-1}\Po\lb S^{0}-S^{\star}\rb+ \lb\Ho^{(-m)}-\Ho\rb\lb L^{\star}\rb \rsb V^{1,m}\rn_{\mathrm{F}}\\
\leq & p^{-1}\ln\lsb\Po^{(-m)}\lb S^{0,m}-S^{0}\rb+\lb\Po^{(-m)}-\Po\rb\lb S^{0}- S^{\star}\rb\rsb V^{1,m}\rn_{\mathrm{F}}
+\ln\lb\Ho-\Ho^{(-m)}\rb\lb L^{\star}\rb V^{1,m}\rn_{\mathrm{F}}\\
= &\underbrace{p^{-1}\ln\lb\Po-\Po^{(-m)}\rb\lb S^{0}- S^{\star}\rb V^{1,m}\rn_{\mathrm{F}}}_{B_1}+\underbrace{\ln\lb\Ho-\Ho^{(-m)}\rb\lb L^{\star}\rb V^{1,m}\rn_{\mathrm{F}}}_{B_2},
\end{align*}
where the second equality holds since $\forall i\in\{0,\cdots,2n\}$, 
$$
S^{0,i} = \T_{\xi^{0}}\lb\Po^{(-i)}\lb M\rb\rb,
$$
therefore $\Po^{(-m)}\lb S^{0,m}-S^{0}\rb$ is a zero matrix.

\noindent\textbf{$\bullet$ Bounding $B_1$.} If $m\leq n$,
\begin{align*}
    B_1 = p^{-1}\ln e_m^T\Po\lb S^0-S^{\star}\rb V^{1,m}\rn_2
    \leq & p^{-1}\cdot(2\alpha pn)\cdot C_{\text{thresh}}\frac{\mu r}{n}\sigma_{1}^{\star}\lb\sqrt{\frac{\mu r}{n}}+\ln \De^{1,\infty} \rn_{2,\infty}\rb \\
    \leq &\frac{1}{2C_0}\frac{\gamma}{\kappa\mu r}\sigma_{r}^{\star}\lb\sqrt{\frac{\mu r}{n}}+\ln \De^{1,\infty} \rn_{2,\infty}\rb,
\end{align*}
where the last inequality follows from the bound in \eqref{eq:S_init_op}. If $m> n$,
\begin{align*}
    B_1 = p^{-1}\ln \Po\lb S^0-S^{\star}\rb e_{(m-n)}e_{(m-n)}^T V^{1,m}\rn_{\mathrm{F}} 
    \leq  & p^{-1}\cdot(2\alpha pn)\cdot C_{\text{thresh}}\frac{\mu r}{n}\sigma_{1}^{\star}\lb\sqrt{\frac{\mu r}{n}}+\ln \De^{1,\infty} \rn_{2,\infty}\rb \\
    \leq & \frac{1}{2C_0}\frac{\gamma}{\kappa\mu r}\sigma_{r}^{\star}\lb\sqrt{\frac{\mu r}{n}}+\ln \De^{1,\infty} \rn_{2,\infty}\rb,
\end{align*}
where the first inequality follows from the fact that $\Po\lb S^0-S^{\star}\rb e_{(m-n)}e_{(m-n)}^T V^{1,m}$ has no more than $2apn$ nonzero rows.

\noindent\textbf{$\bullet$ Bounding $B_2$.} If $m\leq n$,
\begin{align*}
B_2 = \ln e_m^T\Ho\lb L^{\star}\rb V^{1,m}\rn_2\leq\frac{C_{10}}{2c_4C_0}\frac{\gamma}{\kappa\mu r}\sigma_{r}^{\star}\lb\sqrt{\frac{\mu r}{n}} + \ln \De^{1,\infty} \rn_{2,\infty}\rb,
\end{align*}
where the inequality comes from the bound we have already derived in \eqref{eq:vector_bern}. If $m>n$,
\begin{align*}
B_2 = \ln \Ho\lb L^{\star}\rb e_{(m-n)}e_{(m-n)}^T V^{1,m}\rn_{\mathrm{F}}\leq&\ln \Ho\lb L^{\star}\rb e_{(m-n)} \rn_2\ln e_{(m-n)}^T V^{1,m}\rn_2 \\
\leq&\ln \Ho\lb L^{\star}\rb\rn_2\lb\sqrt{\frac{\mu r}{n}} + \ln \De^{1,\infty} \rn_{2,\infty}\rb \\
\leq&\frac{1}{2C_0}\frac{\gamma}{\kappa\mu r}\sigma_{r}^{\star}\lb\sqrt{\frac{\mu r}{n}} + \ln \De^{1,\infty} \rn_{2,\infty}\rb,
\end{align*}
where the last inequality follows from \eqref{eq:op_init}. In both cases,
\begin{equation}\label{eq:WF_init_bound}
\begin{aligned}
\ln W^{0,m} F^{1,m}\rn_{\mathrm{F}}
\leq\frac{1+\frac{C_{10}}{c_4}}{C_0}\frac{\gamma}{\kappa\mu r}\sigma_r^{\star}\lb\sqrt{\frac{\mu r}{n}} + \ln \De^{1,\infty} \rn_{2,\infty}\rb,
\end{aligned}
\end{equation}
and consequently,
\begin{equation}\label{eq:Frobenius}
\begin{aligned}
\ln D^{1,\infty} \rn_{\mathrm{F}} \leq 2\cdot\max_m\ln D^{1,0,m} \rn_{\mathrm{F}}
\leq\frac{4+4\frac{C_{10}}{c_4}}{C_0}\frac{\gamma}{\kappa\mu r}\lb\sqrt{\frac{\mu r}{n}} + \ln \De^{1,\infty} \rn_{2,\infty}\rb.
\end{aligned}
\end{equation}
Combining \eqref{eq:Frobenius} with \eqref{eq:2_infty}, we finally get
\begin{align*}
    \ln \De^{1,\infty} \rn_{2,\infty}
    &\leq \frac{17+\frac{C_{10}}{c_4}}{C_0}\lb\frac{1}{\mu r}\sqrt{\frac{\mu r}{n}}\rb\gamma
    +\frac{1+\frac{C_{10}}{c_4}}{C_0}\ln\De^{1,\infty}\rn_{2,\infty}\\
    &+\frac{1+\frac{C_{10}}{c_4}}{C_0}\frac{\gamma}{\kappa\mu r}\lb\sqrt{\frac{\mu r}{n}} + \ln \De^{1,\infty} \rn_{2,\infty}\rb.
\end{align*}
Using the fact that $c_4>1$ from Lemma \ref{lem:init},
\begin{align*}
    \lb 1-\frac{2+2C_{10}}{C_0}\rb\ln \De^{1,\infty} \rn_{2,\infty}
    \leq &\frac{18+2C_{10}}{C_0}\cdot\lb\frac{1}{\mu r}\sqrt{\frac{\mu r}{n}}\rb\gamma,\\
    \ln \De^{1,\infty} \rn_{2,\infty} \leq & \frac{18+2C_{10}}{C_0-2-2C_{10}}\cdot\lb\frac{1}{\mu r}\sqrt{\frac{\mu r}{n}}\rb\gamma.
\end{align*} 
When $C_{10}\geq 9$, and $C_0\geq 5(2+2C_{10})$, we obtain the bound
$$
\ln \De^{1,\infty} \rn_{2,\infty} \leq \frac{5C_{10}}{C_0}\cdot\lb\frac{1}{\mu r}\sqrt{\frac{\mu r}{n}}\rb\gamma.
$$
Furthermore, 
$$ 
\ln D^{1,\infty} \rn_{\mathrm{F}}
\leq 4\frac{1+C_{10}}{C_0}\frac{\gamma}{\kappa\mu r}\sqrt{\frac{\mu r}{n}}\lb 1 + \frac{5C_{10}}{C_0}\rb\leq \frac{8C_{10}}{C_0}\cdot\lb\frac{1}{\kappa\mu r}\sqrt{\frac{\mu r}{n}}\rb\gamma.
$$

\section{Induction Steps}\label{sec:induction}

{\em In the induction steps, our proofs and the ones in \textup{\cite{Ding2020}} differ mainly in the following two aspects: 
\begin{itemize} 
\item Instead of first bounding $\ln \De^{t+1,\infty}\rn_{2,\infty}$ and then bounding $\ln D^{t+1,\infty}\rn_{\mathrm{F}}$ based on $\ln \De^{t+1,\infty}\rn_{2,\infty}$, we have established two inequalities \eqref{eq:De_induction} and \eqref{eq:D_induction}  which involve $\ln \De^{t+1,\infty}\rn_{2,\infty}$ and  $\ln D^{t+1,\infty}\rn_{\mathrm{F}}$  simultaneously. This allows us to relax the requirement on the distance to the ground truth. Through further refined analysis, in particular by using Lemma {\ref{lem:bound1}} to bound $\tilde{T}_{2a},~\tilde{T}_{2b}$, and $\tilde{B}_{1},~\tilde{B}_{2}$, and using Lemma \ref{lem:bound2} to derive the bounds for $\tilde{T}_{2c}$ and $\tilde{B}_{3}$, we have improved the sample complexity substantially.
\item The outlier term $\ln\Po^{(-m)}\lb S^{t,m}-S^{t}\rb V^{t+1,m}\rn_{\mathrm{F}}$ needs to be bounded carefully, with the help of the Lipschitz property of the thresholding function and Lemma \ref{lem:P_Omega_AB}.
\end{itemize}}
\noindent
Suppose that \Cref{eq:op,eq:l_2_infty,eq:proximity} hold for the $t$-th ($t\geq 1$) iteration. In the sequel, we will show that they also hold for the $(t+1)$-th iteration.

\noindent\textbf{Step 1: The Bound For $\ln E^{t,\infty} \rn_2$.} Using Lemma \ref{lem:L_infinity} with the bounds \eqref{eq:op} and \eqref{eq:l_2_infty},
\begin{align*}
\ln L^{t,i}-L^{\star} \rn_{\infty} \leq & \|\Delta^{t,i}\|_{2, \infty}\lb\|F^{t,i}\|_{2, \infty}+\|F^{\star}\|_{2, \infty}\rb\|\Sigma^{t,i}\|_{2}+(3+4\kappa)\left\|F^{\star}\right\|_{2, \infty}^2\|E^{t-1,i}\|_{2} \\
\leq & \frac{5C_{10}}{C_0}\lb\frac{1}{\mu r}\sqrt{\frac{\mu r}{n}}\rb\gamma^t\cdot\lb 2+\frac{5C_{10}}{C_0}\rb\sqrt{\frac{\mu r}{n}}\cdot\lb 1+\frac{1}{C_0}\rb\sigma_1^{\star}+\frac{\mu r}{n}\cdot\frac{7}{C_0}\lb\frac{1}{\mu r}\sigma_{r}^{\star}\rb\gamma^t\\
\leq & \frac{12C_{10}}{C_0}\cdot\frac{1}{\mu r}\lb\frac{\mu r}{n}\sigma_1^{\star}\rb\gamma^t\leq\lb\frac{\mu r}{n}\sigma_1^{\star}\rb\gamma^t.
\end{align*}
It follows from Lemma~\ref{lem:thresh} that
$$
\text{Supp}\lb S^{t,i}\rb\subseteq\Omega^{(-i)}\cap\Omega_{S^{\star}},~\ln\Po^{(-i)} \lb S^{t,i}-S^{\star}\rb\rn_{\infty}\leq C_{\text{thresh}}\cdot\lb\frac{\mu r}{n}\sigma_{1}^{\star}\rb\gamma^t.
$$
Furthermore, 
\begin{align*}
\ln E_1^{t,i} \rn_2 = \ln p^{-1}\Po^{(-i)}\lb S^{t,i}-S^{\star}\rb\rn_2\leq p^{-1}\cdot\lb2\alpha p n\rb \cdot C_{\text{thresh}}\lb\frac{\mu r}{n}\sigma_{1}^{\star}\rb\gamma^t\leq \frac{1}{2C_0}\cdot\lb\frac{1}{\kappa\mu r}\sigma_{r}^{\star}\rb\gamma^{t+1}.
\end{align*}
For $\ln E_2^{t,i} \rn_2$, one has
\begin{equation}\label{eq:induct_op}
\begin{aligned}
\ln E_2^{t,i} \rn_2 = \ln\Ho^{(-i)}\lb L^{t,i}-L^{\star}\rb\rn_2 \leq & \ln\Ho\lb L^{t,i}-L^{\star}\rb\rn_2\\
\leq & c_5\sqrt{\frac{2rn\log n}{p}}\ln L^{t,i}-L^{\star} \rn_{\infty} \\
\leq & c_5\sqrt{\frac{2rn\log n}{p}}\cdot\frac{12C_{10}}{C_0}\frac{\sigma_1^{\star}}{n}\gamma^t
\leq \frac{1}{2C_0}\cdot\lb\frac{1}{\kappa\mu r}\sigma_{r}^{\star}\rb\gamma^{t+1},
\end{aligned}
\end{equation}
where the second inequality follows from applying Lemma \ref{lem:uniform} to $L^{t,i}-L^{\star}$ which holds for some universal constant $c_5>1$, and the last inequality holds provided that 
$$
p\geq\frac{1152c_5^2C_{10}^2}{\gamma^2}\cdot\frac{\kappa^4\mu^2r^3\log n}{n}.
$$
Therefore,
\begin{align*}
\ln E^{t,i} \rn_2 \leq \frac{1}{C_0}\cdot\lb\frac{1}{\kappa\mu r}\sigma_{r}^{\star}\rb\gamma^{t+1}.
\end{align*}

\noindent\textbf{Step 2: The Bound For $\ln \De^{t+1,\infty} \rn_{2,\infty}$.} We can follow the same argument as in the base case, and split the estimation into four terms. 
\begin{align*}
& \De^{t+1,i}_{m,:} \\
= &  e_m^T\lb F^{t+1,i}-F^{\star}G^{t+1,i}\rb \\
= &\underbrace{e_m^TF^{\star}\Si^{\star}\lsb H^{t+1,i}\lb\Si^{\star}\rb^{-1}-\lb\Si^{\star}\rb^{-1}G^{t+1,i}\rsb}_{T_1}+\underbrace{e_m^TF^{\star}\Si^{\star}H^{t+1,i}\lsb\lb\Si^{t+1,i}\rb^{-1}-\lb\Si^{\star}\rb^{-1}\rsb}_{T_2}\\
&-\frac12\underbrace{e_m^T\widetilde{F}^{\star}\Si^{\star}\lb \widetilde{F}^{\star}\rb^T F^{t+1,i}\lb\Si^{t+1,i}\rb^{-1}}_{T_3}+\underbrace{e_m^T\widehat{E^{t,i}} F^{t+1,i}\lb\Si^{t+1,i}\rb^{-1}}_{T_4}.
\end{align*}

The bounds for the first three terms can be obtained similarly as in the base case, i.e.,
\begin{equation}\label{eq:bounds}
\begin{aligned}
\ln T_1 \rn_2+\ln T_2 \rn_2+\frac12\ln T_3 \rn_2\leq \frac{16}{C_0}\cdot\lb\frac{1}{\mu r}\sqrt{\frac{\mu r}{n}}\rb\gamma^{t+1}.
\end{aligned}
\end{equation}
To get the bound for $\ln T_4 \rn_2$, consider the case that $0\leq i\leq n,~1\leq m\leq n,~i\neq m$. Next we derive a bound for
$$
\ln e_m^TE^{t,i}V^{t+1,i} \rn_2.
$$
The proof for the other cases can be done similarly. While $e_m^TE_1^{t,i} V^{t+1,i}$ can be bounded in the same way as the base case, i.e.,
\begin{align*}
    \ln e_m^TE_1^{t,i} V^{t+1,i}\rn_2\leq \frac{\sigma_{r}^{\star}}{2C_0}\cdot\lb\frac{1}{\kappa\mu r}\gamma^{t+1}\rb\lb\sqrt{\frac{\mu r}{n}}+\ln\De^{t+1,\infty}\rn_{2,\infty}\rb,
\end{align*}
bounding $e_m^TE_2^{t,i} V^{t+1,i}$ is technically quite involved. First, one has
\begin{align*}
    e_m^T E_2^{t, i} V^{t+1, i} = & e_m^T E_2^{t, i} V^{t+1, m} G^{t+1,i,m} +e_m^T E_2^{t, i}\left(V^{t+1, i}-V^{t+1, m} G^{t+1,i,m}\right) \\ = & e_m^T \mathcal{H}_{\Omega}^{(-i)}\left(L^{t, i}-L^{\star}\right) V^{t+1, m} G^{t+1,i,m} +e_m^T E_2^{t, i}\left(V^{t+1, i}-V^{t+1, m} G^{t+1,i,m}\right) \\ 
    = & \underbrace{e_m^T \mathcal{H}_{\Omega}^{(-i)}\left(L^{t, m}-L^{\star}\right) V^{t+1, m} G^{t+1,i,m}}_{\tilde{T}_1} +\underbrace{e_m^T \mathcal{H}_{\Omega}^{(-i)}\left(L^{t, i}-L^{t, m}\right) V^{t+1, m} G^{t+1,i,m}}_{\tilde{T}_2} \\ 
    & +\underbrace{e_m^T E_2^{t, i}\left(V^{t+1, i}-V^{t+1, m} G^{t+1,i,m}\right)}_{\tilde{T}_3}.
\end{align*}

\noindent\textbf{$\bullet$ Bounding $\tilde{T}_1$.} Note that
\begin{align*}
    \ln e_m^T \mathcal{H}_{\Omega}^{(-i)}\left(L^{t, m}-L^{\star}\right) V^{t+1, m} G^{t+1,i,m} \rn_2 = \ln e_m^T \mathcal{H}_{\Omega}^{(-i)}\left(L^{t, m}-L^{\star}\right) V^{t+1, m}\rn_2.
\end{align*}
For $k=1,\cdots,n$, define
$$
v_k = \lb 1-\delta_{mk}/p\rb \lb L_{mk}^{t, m}-L_{mk}^{\star}\rb V_{k,:}^{t+1,m}.
$$
Similar to \eqref{eq:vector_bern}, we can apply Lemma \ref{lem:bernstein} to get
\begin{equation}\label{eq:I_3} 
\begin{aligned}
    & \ln e_m^T \mathcal{H}_{\Omega}^{(-i)}\left(L^{t, m}-L^{\star}\right) V^{t+1, m} \rn_2 = \ln \sum_{k=1}^n v_k\rn_2\\
    \leq & C_{10}\cdot\lb\sqrt{\frac{\log n}{p}}\ln L^{t, m}- L^{\star} \rn_{2,\infty}+\frac{\log n}{p}\ln L^{t, m}-L^{\star} \rn_{\infty}\rb\cdot\ln V^{t+1,m}\rn_{2,\infty} \\
    \leq & C_{10}\cdot\lb\sqrt{\frac{n\log n}{p}}+\frac{\log n}{p}\rb\ln L^{t, m}-L^{\star} \rn_{\infty}\cdot\lb\sqrt{\frac{\mu r}{n}}+\ln \De^{t+1,\infty} \rn_{2,\infty}\rb \\
    \leq & C_{10}\cdot\lb\sqrt{\frac{n\log n}{p}}+\frac{\log n}{p}\rb\cdot\frac{12C_{10}}{C_0}\frac{\sigma_1^{\star}}{n}\gamma^t\cdot\lb\sqrt{\frac{\mu r}{n}} + \ln \De^{t+1,\infty} \rn_{2,\infty}\rb\\
    \leq & \frac{C_{10}}{C_0}\cdot\lb\frac{\sigma_r^{\star}}{\mu r}\gamma^{t+1}\rb\lb\sqrt{\frac{\mu r}{n}} + \ln \De^{t+1,\infty} \rn_{2,\infty}\rb,
\end{aligned}
\end{equation}
where the last inequality holds provided that 
$$
p\geq\frac{576C_{10}^2}{\gamma^2}\cdot\frac{\kappa^2\mu^2r^2\log n}{n}.
$$

\noindent\textbf{$\bullet$ Bounding $\tilde{T}_2$.} 
Note that
\begin{equation}\label{eq:decomposition}
\begin{aligned}  
L^{t, m}-L^{t, i} 
= & U^{t, m} \Si^{t, m}\lb V^{t, m}\rb^T-U^{t, i} \Si^{t, i} \lb V^{t, i}\rb^T \\ 
= & U^{t, m} \Si^{t, m} \lb V^{t, m}\rb^T - U^{t, i} G^{t,m,i}\Si^{t, m} \lb V^{t, m}\rb^T   \\ 
&+U^{t, i} G^{t,m,i}\Si^{t, m} \lb V^{t, m}\rb^T-
U^{t, i} \Si^{t, i} G^{t,m,i}\lb V^{t, m}\rb^T
\\ 
& + U^{t, i} \Si^{t, i} G^{t,m,i}\lb V^{t, m}\rb^T-U^{t, i} \Si^{t, i} \lb V^{t, i}\rb^T\\ 
= & D_U^{t, m, i} \Si^{t, m} \lb V^{t, m}\rb^T+U^{t,i} S^{t, m,i}\lb V^{t, m}\rb^T - U^{t, i} \Si^{t, i} \lb D_V^{t, i, m}\rb^T,
\end{aligned}
\end{equation}
where
\begin{align*}
D_U^{t,m,i} &:= U^{t,m}-U^{t,i}G^{t,m,i},\\
S^{t,m,i} &:= G^{t,m,i}\Sigma^{t,m}-\Sigma^{t,i}G^{t,m,i},\\
D_V^{t,i,m} &:= V^{t,i}-V^{t,m}G^{t,i,m}= V^{t,i}-V^{t,m}\lb G^{t,m,i}\rb^T.
\end{align*}
Therefore,
\begin{align*}  
\ln\tilde{T_2}\rn_2\leq &\underbrace{\ln e_m^T\Ho^{(-i)}\lb D_U^{t, m, i} \Si^{t, m} \lb V^{t, m}\rb^T\rb V^{t+1,m}\rn_2}_{\tilde{T}_{2a}}+\underbrace{\ln e_m^T\Ho^{(-i)}\lb U^{t,i} S^{t,m,i}\lb V^{t,m}\rb^T\rb V^{t+1,m}\rn_2}_{\tilde{T}_{2b}} \\
&+ \underbrace{\ln e_m^T\Ho^{(-i)}\lb U^{t, i} \Si^{t, i} \lb D_V^{t, i, m}\rb^T\rb V^{t+1,m}\rn_2}_{\tilde{T}_{2c}}.
\end{align*}

Applying Lemma \ref{lem:bound1} with
$A:=D_U^{t, m, i} \Si^{t, m}$, $B:=V^{t,m}$, and $C:=V^{t+1,m}$, we can get
\begin{align*} 
    \ln \tilde{T}_{2a}\rn_2 \leq \ln D_U^{t, m, i} \Si^{t, m} \rn_{2,\infty}\cdot \ln R^{(m)}\rn_2,
\end{align*}
where
\begin{equation}\label{eq:bernstein}
R^{(m)} := \sum_{k=1}^n \lb 1- \delta_{mk}/p\rb \lb V^{t,m}_{k,:}\rb^T V^{t+1,m}_{k,:}:= \sum_{k=1}^n R_{k}^{(m)}.
\end{equation}
 Since
\begin{align*}
    \ln R_{k}^{(m)} \rn_2\leq & \frac1p \ln V^{t,m} \rn_{2,\infty}\ln V^{t+1,m} \rn_{2,\infty}, \\
    \ln \sum_{k=1}^n\mathbb{E} \lsb R_k^{(m)}\lb R_k^{(m)}\rb^T\rsb\rn_2 \leq & \frac1p\ln\sum_{k=1}^n  \lsb \ln V_{k,:}^{t+1,m}\rn_2^2\cdot\lb V_{k,:}^{t,m} \rb^T V_{k,:}^{t,m}\rsb \rn_2
    \leq \frac1p\ln V^{t+1,m} \rn_{2,\infty}^2,\\
    \ln \sum_{k=1}^n\mathbb{E} \lsb \lb R_k^{(m)}\rb^T R_k^{(m)}\rsb\rn_2 \leq & \frac1p\ln\sum_{k=1}^n \lsb\ln V_{k,:}^{t,m}\rn_2^2\cdot\lb V_{k,:}^{t+1,m} \rb^T V_{k,:}^{t+1,m} \rsb\rn_2
    \leq \frac1p\ln V^{t,m} \rn_{2,\infty}^2,
\end{align*}
applying Lemma \ref{lem:bernstein} and we can obtain
\begin{align*}
    \ln  R^{(m)} \rn_2 \leq & C_{10}\cdot\lb\sqrt{\frac{\log n}{p}}\cdot\max\left\{\ln V^{t+1,m} \rn_{2,\infty},\ln V^{t,m} \rn_{2,\infty}\right\}+\frac{\log n}{p}\ln V^{t,m} \rn_{2,\infty}\ln V^{t+1,m} \rn_{2,\infty}\rb\\
    \leq & C_{10}\cdot\sqrt{\frac{\log n}{p}}\lb 1+\frac{5C_{10}}{C_0}\rb\lb\sqrt{\frac{\mu r}{n}}+\ln \De^{t+1,\infty}\rn_{2,\infty}\rb\\
    &+C_{10}\cdot\frac{\log n}{p}\lb 1+\frac{5C_{10}}{C_0}\rb\sqrt{\frac{\mu r}{n}}\lb\sqrt{\frac{\mu r}{n}}+\ln \De^{t+1,\infty}\rn_{2,\infty}\rb \\
    = & C_{10}\lb 1+\frac{5C_{10}}{C_0}\rb\cdot\lb\sqrt{\frac{\log n}{p}}+\frac{\log n}{p}\sqrt{\frac{\mu r}{n}}\rb\cdot\lb\sqrt{\frac{\mu r}{n}}+\ln\De^{t+1,\infty}\rn_{2,\infty}\rb\\
    =& C_{10}\lb 1+\frac{5C_{10}}{C_0}\rb\lb\sqrt{\frac{\mu r\log n}{np}}+\frac{\mu r\log n}{np}\rb\cdot\sqrt{\frac{n}{\mu r}}\lb\sqrt{\frac{\mu r}{n}}+\ln\De^{t+1,\infty}\rn_{2,\infty}\rb \\
    \leq & \frac{\gamma}{9}\cdot\sqrt{\frac{n}{\mu r}}\lb\sqrt{\frac{\mu r}{n}}+\ln\De^{t+1,\infty}\rn_{2,\infty}\rb,
\end{align*}
where the last inequality holds if 
$$
p\geq\frac{400C_{10}^2}{\gamma^2}\cdot\frac{\mu r\log n}{n}.
$$
Therefore,
\begin{align*} 
    \ln \tilde{T}_{2a}\rn_2
    \leq &\frac{8C_{10}}{C_0}\lb\frac{1}{\kappa\mu r}\sqrt{\frac{\mu r}{n}}\rb\gamma^t\cdot\lb 1+\frac{1}{C_0}\rb\sigma_1^{\star}\cdot\frac{\gamma}{9}\sqrt{\frac{n}{\mu r}}\lb\sqrt{\frac{\mu r}{n}}+\ln\De^{t+1,\infty}\rn_{2,\infty}\rb \\
    \leq &\frac{C_{10}}{C_0}\cdot\lb\frac{\sigma_r^{\star}}{\mu r}\gamma^{t+1}\rb\lb\sqrt{\frac{\mu r}{n}}+\ln\De^{t+1,\infty}\rn_{2,\infty}\rb.
\end{align*}

Since $\tilde{T}_{2b}$ can be bounded similarly as $\tilde{T}_{2a}$ (by replacing $\ln D_U^{t, m, i} \Si^{t, m} \rn_{2,\infty}$ with $\ln U^{t,i} S^{t,m,i} \rn_{2,\infty}$),  we just need to bound $\ln S^{t,m,i} \rn_2$ next. Let $A:=\widehat{L^{\star}}+\widehat{E^{t-1,m}}$, $\widetilde{A}:=\widehat{L^{\star}}+\widehat{E^{t-1,i}}$. The application of Lemma \ref{lem:perturb_S} yields
$$
\ln S^{t,m,i} \rn_2=\ln \Sigma^{t,m}G^{t,i,m}-G^{t,i,m}\Sigma^{t,i}\rn_2\leq 4\kappa\cdot\ln E^{t-1,i}-E^{t-1,m}\rn_2 \leq\frac{8}{C_0}\cdot\lb\frac{1}{\mu r}\sigma_{r}^{\star}\rb\gamma^{t}.
$$
With the same requirement for $p$ in bounding $\tilde{T}_{2a}$, we can get
\begin{align*}
    \tilde{T}_{2b}\leq\frac{2}{C_0}\cdot\lb\frac{\sigma_r^{\star}}{\mu r}\gamma^{t+1}\rb\lb\sqrt{\frac{\mu r}{n}} + \ln \De^{t+1,\infty} \rn_{2,\infty}\rb.
\end{align*}

For $\tilde{T}_{2c}$,
\begin{align*}
\tilde{T}_{2c}^2 = & \ln e_m^T\Ho\lb U^{t,i}\Si^{t,i} \lb D_V^{t, i, m}\rb^T \rb V^{t+1,m}\rn_2^2 \\
= & \left\langle e_m^T\Ho\lb U^{t,i}\Si^{t,i} \lb D_V^{t, i, m}\rb^T \rb V^{t+1,m},e_m^T\Ho\lb U^{t,i}\Si^{t,i} \lb D_V^{t, i, m}\rb^T \rb V^{t+1,m}\right\rangle \\
= & \left\langle \Ho\lb U^{t,i}\Si^{t,i} \lb D_V^{t, i, m}\rb^T \rb ,e_me_m^T\Ho\lb U^{t,i}\Si^{t,i} \lb D_V^{t, i, m}\rb^T \rb V^{t+1,m}\lb V^{t+1,m}\rb^T\right\rangle \\
\leq & \ln\Ho\lb \bm{1}\bm{1}^T\rb\rn_2 \ln U^{t,i}\Si^{t,i} \rn_{2,\infty}\ln V^{t+1,m} \rn_{2,\infty}\ln D_V^{t, i, m} \rn_{\mathrm{F}}\cdot\tilde{T}_{2c},
\end{align*}
where the inequality follows from applying Lemma~\ref{lem:bound2} to
$$
A:=U^{t,i}\Si^{t,i},~C:=D_V^{t, i, m},~B:=e_me_m^T\Ho\lb U^{t,i}\Si^{t,i} \lb D_V^{t, i, m}\rb^T\rb V^{t+1,m},~D:=V^{t+1,m}.
$$ 
Therefore,
\begin{align*}
\tilde{T}_{2c} \leq & \ln\Ho\lb \bm{1}\bm{1}^T\rb\rn_2 \ln U^{t,i}\Si^{t,i} \rn_{2,\infty}\ln D_V^{t, i, m} \rn_{\mathrm{F}}\ln V^{t+1,m} \rn_{2,\infty}\\
\leq & c_5\sqrt{\frac{n\log n}{p}}\cdot\lb 1+\frac{5C_{10}}{C_0}\rb\sqrt{\frac{\mu r}{n}}\lb 1+\frac{1}{C_0}\rb\sigma_1^{\star}\cdot\frac{8C_{10}}{C_0}\lb\frac{1}{\kappa\mu r}\sqrt{\frac{\mu r}{n}}\rb\gamma^{t}\cdot\ln V^{t+1,m} \rn_{2,\infty} \\
\leq & \frac{1}{C_0}\cdot\lb\frac{\sigma_r^{\star}}{\mu r}\gamma^{t+1}\rb\lb\sqrt{\frac{\mu r}{n}} + \ln \De^{t+1,\infty} \rn_{2,\infty}\rb,
\end{align*}
where the second inequality follows from applying Lemma \ref{lem:uniform} to $\bm{1}\bm{1}^T$, and the last inequality holds provided that 
$$
p\geq\frac{100c_5^2C_{10}^2}{\gamma^2}\cdot\frac{\mu^2 r^2\log n}{n}.
$$

\noindent\textbf{$\bullet$ Bounding $\tilde{T}_3$.} \begin{align*}
    \ln\tilde{T}_3\rn_2 = & \ln e_m^T E_2^{t, i}\lb V^{t+1, i}-V^{t+1, m} G^{t+1,i,m}\rb\rn_2
    \leq \ln E_2^{t,i}\rn_2\ln D_V^{t+1,i,m}\rn_{\mathrm{F}}\leq \frac{1}{2C_0}\sigma_{r}^{\star}\ln D^{t+1,\infty}\rn_{\mathrm{F}}.
\end{align*}

Combining the bounds of $\tilde{T}_1$ to $\tilde{T}_3$, one has
\begin{align*}
    \ln e_m^TE_2^{t,i} V^{t+1,i} \rn_2
    \leq & \frac{3+2C_{10}}{C_0}\lb\frac{\sigma_r^{\star}}{\mu r}\sqrt{\frac{\mu r}{n}}\rb\gamma^{t+1} + \frac{3+2C_{10}}{C_0}\sigma_r^{\star}\ln \De^{t+1,\infty} \rn_{2,\infty}+\frac{1}{2C_0}\sigma_r^{\star}\ln D^{t+1,\infty}\rn_{\mathrm{F}},
\end{align*}
and consequently,
\begin{align*}
    \ln T_4 \rn_2
    \leq & \frac{4+3C_{10}}{C_0}\lb\frac{1}{\mu r}\sqrt{\frac{\mu r}{n}}\rb\gamma^{t+1} + \frac{4+3C_{10}}{C_0}\ln \De^{t+1,\infty} \rn_{2,\infty}+\frac{1}{C_0}\ln D^{t+1,\infty}\rn_{\mathrm{F}}.
\end{align*}

Together with \eqref{eq:bounds}, we can get
\begin{equation}\label{eq:De_induction}
    \lb 1-\frac{4+3C_{10}}{C_0}\rb\ln \De^{t+1,\infty} \rn_{2,\infty}
    \leq \frac{20+3C_{10}}{C_0}\lb\frac{1}{\mu r}\sqrt{\frac{\mu r}{n}}\rb\gamma^{t+1}+\frac{1}{C_0}\ln D^{t+1,\infty}\rn_{\mathrm{F}}.
\end{equation}

\noindent\textbf{Step 3: The Bound For $\ln D^{t+1,\infty} \rn_{\mathrm{F}}$.} Similar to the base case, we only need to bound $\ln D^{t+1,0,m}\rn_{\mathrm{F}}$. Let $A:=\widehat{L^{\star}}+\widehat{E^{t,m}}$, $\widetilde{A}:=\widehat{L^{\star}}+\widehat{E^{t}}$, and define $W^{t,m}:=\widehat{E^{t}}-\widehat{E^{t,m}}$. By Lemma \ref{lem:perturb_F}, 
\begin{equation}\label{eq:WF_induction}
\begin{aligned}
\ln D^{t+1,0,m}\rn_{\mathrm{F}} \leq &\frac{2}{\sigma_r^{\star}}\ln W^{t,m} F^{t+1,m}\rn_{\mathrm{F}} \\
\leq &\frac{2}{\sigma_r^{\star}}\lb \ln\lb E^{t}-E^{t,m}\rb V^{t+1,m}\rn_{\mathrm{F}}+\ln\lb E^{t}-E^{t,m}\rb^T U^{t+1,m}\rn_{\mathrm{F}}\rb.
\end{aligned}
\end{equation}
Next we  derive a bound for $\ln\lb E^{t}-E^{t,m}\rb V^{t+1,m}\rn_{\mathrm{F}}$, and the bound for $\ln\lb E^{t}-E^{t,m}\rb^T U^{t+1,m}\rn_{\mathrm{F}}$ can be obtained similarly. Note that
\begin{align*}
& \ln\lb E^{t}-E^{t,m}\rb V^{t+1,m}\rn_{\mathrm{F}} \\
\leq &\ln \lsb p^{-1}\Po^{(-m)}\lb S^{t,m}-S^{\star}\rb-p^{-1}\Po\lb S^{t}-S^{\star}\rb\rsb V^{t+1,m}\rn_{\mathrm{F}}\\
&+\ln \lsb\Ho\lb L^{t}-L^{\star}\rb-\Ho^{(-m)}\lb L^{t,m}-L^{\star}\rb\rsb V^{t+1,m}\rn_{\mathrm{F}}\\
\leq & \underbrace{p^{-1}\ln\Po^{(-m)}\lb S^{t,m}-S^{t}\rb V^{t+1,m}\rn_{\mathrm{F}}}_{B_1}+\underbrace{p^{-1}\ln \lb\Po^{(-m)}-\Po\rb\lb S^{t}-S^{\star}\rb V^{t+1,m}\rn_{\mathrm{F}}}_{B_2}\\
&+\underbrace{\ln\Ho\lb L^{t}-L^{t,m}\rb V^{t+1,m}\rn_{\mathrm{F}}}_{B_3}+\underbrace{\ln\lb\Ho-\Ho^{(-m)}\rb\lb L^{t,m}-L^{\star}\rb V^{t+1,m}\rn_{\mathrm{F}}}_{B_4}.
\end{align*}

\noindent\textbf{$\bullet$ Bounding $B_1$.} Let $\Omega_{S^{\star}}^{(-m)}$  be $\Omega_{S^{\star}}$ without the indices from the $m$-th row if $1\leq m\leq n$, and $\Omega_{S^{\star}}$ without the indices from the $(m-n)$-th column if $n+1\leq m\leq 2n$. For $(i,j)\in[n]\times[n]$, define $w_{ij}=1$ if $(i,j)\in\Omega_{S^{\star}}^{(-m)}$ and $w_{ij}=0$  otherwise. Then,
\begin{align*}
    \ln\Po^{(-m)}\lb S^{t,m}-S^{t}\rb V^{t+1,m}\rn_{\mathrm{F}}^2 
    = & \sum_{i=1}^n\sum_{j=1}^r \lsb\sum_{k=1}^n w_{ik}\lb S^{t,m}-S^{t}\rb_{ik} V^{t+1,m}_{kj}\rsb^2 \\
    \leq & \sum_{i=1}^n\sum_{j=1}^r \lsb\sum_{k=1}^n w_{ik}\lb S^{t,m}-S^{t}\rb_{ik}^2\rsb \lsb\sum_{k=1}^n w_{ik}\lb V^{t+1,m}_{kj}\rb^2\rsb.
\end{align*}
Note that
\begin{align*}
    \sum_{k=1}^n w_{ik}\lb V^{t+1,m}_{kj}\rb^2 = & \ln\mathcal{P}_{\Omega_{S^{\star}}^{(-m)}}\lb e_ie_j^T\cdot\lb V^{t+1,m}\rb^T\rb\rn_{\mathrm{F}}^2\\
    \leq & \ln\mathcal{P}_{\Omega_{S^{\star}}}\lb e_ie_j^T\cdot\lb V^{t+1,m}\rb^T\rb\rn_{\mathrm{F}}^2\leq 2\alpha pn\cdot\ln V^{t+1,m}\rn_{2,\infty}^2,
\end{align*}
where the last inequality is due to Lemma \ref{lem:P_Omega_AB}. Therefore,
\begin{align*}
    \ln\Po^{(-m)}\lb S^{t,m}-S^{t}\rb V^{t+1,m}\rn_{\mathrm{F}}^2
    \leq & \sum_{i=1}^n\lsb\sum_{k=1}^n w_{ik}\lb S^{t,m}-S^{t}\rb_{ik}^2\rsb\cdot\sum_{j=1}^r\lsb\sum_{k=1}^n w_{ik}\lb V^{t+1,m}_{kj}\rb^2\rsb \\
    \leq & \sum_{i=1}^n\lsb\sum_{k=1}^n w_{ik}\lb S^{t,m}-S^{t}\rb_{ik}^2\rsb\cdot\lb 2 \alpha pn r\ln V^{t+1,m}\rn_{2,\infty}^2\rb\\
    = & \ln\P_{\Omega_{S^{\star}}}^{(-m)}\lb S^{t,m}-S^{t}\rb\rn_{\mathrm{F}}^2\cdot\lb 2\alpha pn r\ln V^{t+1,m}\rn_{2,\infty}^2\rb.
\end{align*}

On the other hand, the term $\ln\P_{\Omega_{S^{\star}}}^{(-m)}\lb S^{t,m}-S^{t}\rb\rn_{\mathrm{F}}$ can be bounded as follows:
\begin{align*}
    &\ln\P_{\Omega_{S^{\star}}}^{(-m)}\lb S^{t,m}-S^{t}\rb\rn_{\mathrm{F}}\\
    = & \ln\P_{\Omega_{S^{\star}}}^{(-m)}\lb \T_{\xi^{t}}\lb M-L^{t,m}\rb-\T_{\xi^{t}}\lb M-L^{t}\rb\rb\rn_{\mathrm{F}} \\
    \leq & K\cdot\ln\P_{\Omega_{S^{\star}}}^{(-m)}\lb\lb M-L^{t,m}\rb-\lb M-L^{t}\rb\rb\rn_{\mathrm{F}} \\
    \leq & K\ln\P_{\Omega_{S^{\star}}}\lb L^{t,m}-L^{t}\rb\rn_{\mathrm{F}} \\
    \leq & K\ln \P_{\Omega_{S^{\star}}}\lb D_U^{t,m,0}\Si^{t,m}\lb V^{t,m}\rb^T\rb\rn_{\mathrm{F}}+K
    \ln \P_{\Omega_{S^{\star}}}\lb U^{t}S^{t,m,0}\lb V^{t,m}\rb^T\rb\rn_{\mathrm{F}}\\
    &+K\ln \P_{\Omega_{S^{\star}}}\lb U^{t}\Si^{t}\lb D_V^{t,0,m}\rb^T\rb\rn_{\mathrm{F}}\\
    \leq & K\sqrt{2\alpha pn}\cdot\lb \ln D_U^{t,m,0}\Si^{t,m}\rn_{\mathrm{F}}\ln V^{t,m}\rn_{2,\infty}+\ln U^{t}S^{t,m,0}\rn_{\mathrm{F}}\ln V^{t,m}\rn_{2,\infty}+\ln U^{t}\rn_{2,\infty}\ln D_V^{t,0,m}\Si^{t}\rn_{\mathrm{F}}\rb\\
    \leq & K\sqrt{2\alpha pn}\cdot\lb 2\lb 1+\frac{1}{C_0}\rb\sigma_1^{\star}\ln D^{t,m,0} \rn_{\mathrm{F}}+\ln S^{t,m,0} \rn_{\mathrm{F}}\rb\cdot\lb 1+\frac{5C_{10}}{C_0}\rb\sqrt{\frac{\mu r}{n}},
\end{align*}
where the first inequality follows from property \labelcref{P2} of the thresholding function (applied to each entry in $\Omega_{S^{\star}}^{(-m)}$), and the fourth inequality follows from Lemma~\ref{lem:P_Omega_AB} again.

Regarding $\ln S^{t,m,0} \rn_{\mathrm{F}}$, we remind the reader that the bound for $\|D^{t,\infty}\|_{\mathrm{F}}$ is indeed provided through the bound for
$\ln W^{t-1,m} F^{t,m}\rn_{\mathrm{F}}$, see \eqref{eq:WF_init} for the case when $t=1$ and  \eqref{eq:WF_induction} for the case when $t>1$. Therefore, in the $t$-th iteration, we have already proved that
$$
\ln W^{t-1,m} F^{t,m}\rn_{\mathrm{F}}
\leq\frac{2C_{10}}{C_0}\cdot
\lb\frac{\sigma_r^{\star}}{\kappa\mu r}\sqrt{\frac{\mu r}{n}}\rb\gamma^t.
$$
For the base case, such a bound can be established by substituting the bound of $\ln \De^{1,\infty}\rn_{2,\infty}$ into \eqref{eq:WF_init_bound}. For the induction steps, such a bound can be established by substituting the bound of $\ln \De^{t,\infty}\rn_{2,\infty}$ and $\ln D^{t,\infty}\rn_{\mathrm{F}}$ into \eqref{eq:WF_induct_bound} (the numbering is changed from $t+1$ to $t$). Applying Lemma \ref{lem:perturb_S} with $A:=\widehat{L^{\star}}+\widehat{E^{t-1,m}}$ and $\widetilde{A}:=\widehat{L^{\star}}+\widehat{E^{t-1}}$ yields
\begin{equation}\label{eq:S_F}
\ln S^{t,m,0} \rn_{\mathrm{F}} = \ln \Sigma^{t,m}G^{t,0,m}-G^{t,0,m}\Sigma^{t}\rn_{\mathrm{F}}\leq 4\kappa\cdot\ln W^{t-1,m} F^{t,m}\rn_{\mathrm{F}}\leq \frac{8C_{10}}{C_0}\cdot\lb\frac{\sigma_1^{\star}}{\kappa\mu r}\sqrt{\frac{\mu r}{n}}\rb\gamma^{t}.
\end{equation}

Combining the above pieces together, one has
\begin{align*}
    B_1 \leq & p^{-1}\cdot\ln\P_{\Omega_{S^{\star}}}^{(-m)}\lb S^{t,m}-S^{t}\rb\rn_{\mathrm{F}}\cdot\sqrt{2\alpha pnr}\ln V^{t+1,m}\rn_{2,\infty}\\
    \leq & 4K\cdot\lb \lb 2+\frac{2}{C_0}\rb\sigma_1^{\star}\ln D^{t,m,0} \rn_{\mathrm{F}}+\ln S^{t,m,0} \rn_{\mathrm{F}}\rb\cdot\alpha r\sqrt{\mu n}\lb\sqrt{\frac{\mu r}{n}}+\ln\De^{t+1,\infty}\rn_{2,\infty}\rb \\
    \leq & 100K\cdot\frac{C_{10}}{C_0}\lb\frac{\sigma_1^{\star}}{\kappa\mu r}\sqrt{\frac{\mu r}{n}}\rb\gamma^{t}\cdot\alpha r\sqrt{\mu n}\lb\sqrt{\frac{\mu r}{n}}+\ln\De^{t+1,\infty}\rn_{2,\infty}\rb \\
    \leq &\frac{C_{10}}{20C_0}\cdot\lb\frac{\sigma_r^{\star}}{\kappa\mu r}\rb\gamma^{t+1}\lb \sqrt{\frac{\mu r}{n}}+\ln\De^{t+1,\infty}\rn_{2,\infty}\rb
\end{align*}
as long as
$$
\alpha\leq\frac{1}{\kappa\mu r^{3/2}}\cdot\frac{\gamma}{2000K}.
$$

\noindent\textbf{$\bullet$ Bounding $B_2$.} Similar to the base case, we can get
\begin{align*}
    B_2 = p^{-1}\ln e_m^T\Po\lb S^t-S^{\star}\rb V^{t+1,m}\rn_2
\end{align*}
if $m\leq n$, and
\begin{align*}
    B_2 = p^{-1}\ln \Po\lb S^t-S^{\star}\rb e_{(m-n)}e_{(m-n)}^T V^{t+1,m}\rn_{\mathrm{F}} 
\end{align*}
if $m>n$. In both cases,
\begin{align*}
    B_2\leq & p^{-1}\cdot(2\alpha pn)\cdot C_{\text{thresh}}\lb\frac{\mu r}{n}\sigma_{1}^{\star}\rb\gamma^t\cdot\lb\sqrt{\frac{\mu r}{n}}+\ln \De^{t+1,\infty} \rn_{2,\infty}\rb \\
    \leq &\frac{1}{2C_0}\cdot\lb\frac{\sigma_{r}^{\star}}{\kappa\mu r}\rb\gamma^{t+1}\lb\sqrt{\frac{\mu r}{n}}+\ln \De^{t+1,\infty} \rn_{2,\infty}\rb.
\end{align*}

\noindent\textbf{$\bullet$ Bounding $B_3$.} It is easy to see that
\begin{align*}
    B_3 \leq & \underbrace{\ln\Ho\lb D_U^{t, m, 0} \Si^{t, m} \lb V^{t, m}\rb^T\rb V^{t+1,m}\rn_{\mathrm{F}}}_{\tilde{B}_1}+\underbrace{\ln\Ho\lb U^{t} S^{t, m,0}\lb V^{t, m}\rb^T\rb V^{t+1,m}\rn_{\mathrm{F}}}_{\tilde{B}_2} \\
    & + \underbrace{\ln\Ho\lb U^{t}\Si^{t} \lb D_V^{t, 0, m}\rb^T \rb V^{t+1,m}\rn_{\mathrm{F}}}_{\tilde{B}_3}.
\end{align*}

For $\tilde{B}_{1}$, the application of Lemma~\ref{lem:bound1} yields that
\begin{align*}
\tilde{B}_1
\leq & \ln D_U^{t, m, 0} \Si^{t, m} \rn_{\mathrm{F}}\cdot\max_{1\leq j\leq n} \ln R^{(j)} \rn_2,
\end{align*}
where
$$
R^{(j)} := \sum_{k=1}^n \lb 1- \delta_{jk}/p\rb \lb V^{t,m}_{k,:}\rb^T V^{t+1,m}_{k,:}.
$$ 
For each $j\in\{1,\cdots,n\}$, denoting $
H^{(j)} := \diag\lb 1-\delta_{j1}/p,\cdots,1-\delta_{jn}/p\rb\in\mathbb{R}^{n\times n}
$,
\begin{align*}
\ln R^{(j)} \rn_2 = & \ln\lb V^{t, m}\rb^T H^{(j)}V^{t+1,m}\rn_2 \\
= & \ln\lb D_V^{t,m,j}+V^{t, j}G^{t,m,j}\rb^T H^{(j)}\lb D_V^{t+1,m,j}+V^{t+1,j}G^{t+1,m,j}\rb\rn_2 \\
\leq & \ln\lb D_V^{t,m,j}\rb^T H^{(j)} D_V^{t+1,m,j}\rn_2+\ln\lb V^{t, j}\rb^T H^{(j)} D_V^{t+1,m,j}\rn_2\\
&+\ln\lb D_V^{t,m,j}\rb^T H^{(j)}V^{t+1,j}\rn_2+\ln\lb V^{t, j}\rb^T H^{(j)} V^{t+1,j}\rn_2.
\end{align*}
Conditioning on the event that
$
\sum_{k=1}^n \delta_{jk} \leq 2pn,
$
\begin{align*}
\ln\lb D_V^{t,m,j}\rb^T H^{(j)} D_V^{t+1,m,j}\rn_2 = &\ln\sum_{k=1}^n \lb 1- \delta_{jk}/p\rb \lb D_V^{t,m,j}\rb_{k,:}^T  \lb D_V^{t+1,m,j}\rb_{k,:}\rn_2 \\
\leq & 3n \ln D_V^{t,\infty}\rn_{2,\infty}\ln D_V^{t+1,\infty}\rn_{2,\infty} \\
\leq & 3n\cdot\frac{8C_{10}}{C_0}\lb\frac{1}{\kappa\mu r}\sqrt{\frac{\mu r}{n}}\rb\gamma^t\ln D^{t+1,\infty}\rn_{\mathrm{F}}\leq \sqrt{\mu rn}\ln D^{t+1,\infty}\rn_{\mathrm{F}},
\end{align*}
\begin{align*}
\ln\lb V^{t, j}\rb^T H^{(j)} D_V^{t+1,m,j}\rn_2 
\leq & 3n \ln V^{t,j}\rn_{2,\infty}\ln D_V^{t+1,\infty}\rn_{2,\infty} \\
\leq & 3n\cdot\lb 1+\frac{5C_{10}}{C_0}\rb\sqrt{\frac{\mu r}{n}}\ln D^{t+1,\infty}\rn_{\mathrm{F}}\leq 4\sqrt{\mu rn}\ln D^{t+1,\infty}\rn_{\mathrm{F}},
\end{align*}
\begin{align*}
\ln\lb D_V^{t,m,j}\rb^T H^{(j)}V^{t+1,j}\rn_2 
\leq & 3n \ln D_V^{t,\infty}\rn_{2,\infty}\ln V^{t+1,j}\rn_{2,\infty}\\
\leq & 3n\cdot\frac{8C_{10}}{C_0}\lb\frac{1}{\kappa\mu r}\sqrt{\frac{\mu r}{n}}\rb\gamma^t\lb \sqrt{\frac{\mu r}{n}}+\ln \De^{t+1,\infty}\rn_{2,\infty}\rb \\
\leq & \frac{\gamma}{360\kappa}\sqrt{\frac{n}{\mu r}}\lb\sqrt{\frac{\mu r}{n}}+\ln\De^{t+1,\infty}\rn_{2,\infty}\rb
\end{align*}
provided that $C_0\geq 8640C_{10}$. 

Since $V^{t, j}$ and $V^{t+1,j}$ are independent with respect to the random variables on the $j$-th row,
\begin{align*}
\ln\lb V^{t, j}\rb^T H^{(j)} V^{t+1,j}\rn_2 = & \ln\sum_{k=1}^n \lb 1- \delta_{jk}/p\rb \lb V^{t,j}_{k,:}\rb^T  V^{t+1,j}_{k,:}\rn_2\\
\leq &\frac{\gamma}{360\kappa}\sqrt{\frac{n}{\mu r}}\lb\sqrt{\frac{\mu r}{n}}+\ln\De^{t+1,\infty}\rn_{2,\infty}\rb,
\end{align*}
where the last bound follows from the same argument for \eqref{eq:bernstein}, provided that
$$
p \geq \frac{640000C_{10}^2}{\gamma^2}\cdot\frac{\kappa^2\mu r\log n}{n}.
$$

Therefore,
\begin{align*}
    \tilde{B}_1\leq & \ln D_U^{t, m, 0} \Si^{t, m} \rn_{\mathrm{F}}\lsb \frac{\gamma}{180\kappa}\sqrt{\frac{n}{\mu r}}\lb\sqrt{\frac{\mu r}{n}}+\ln\De^{t+1,\infty}\rn_{2,\infty}\rb+5\sqrt{\mu rn}\ln D^{t+1,\infty}\rn_{\mathrm{F}}\rsb \\
    \leq & \frac{8C_{10}}{C_0}\lb\frac{1}{\kappa\mu r}\sqrt{\frac{\mu r}{n}}\rb\gamma^t\cdot\lb 1+\frac{1}{C_0}\rb\sigma_1^{\star}\lsb \frac{\gamma}{180\kappa}\sqrt{\frac{n}{\mu r}}\lb\sqrt{\frac{\mu r}{n}}+\ln\De^{t+1,\infty}\rn_{2,\infty}\rb+5\sqrt{\mu rn}\ln D^{t+1,\infty}\rn_{\mathrm{F}}\rsb \\
    \leq & \frac{C_{10}}{20C_0}\lb\frac{\sigma_r^{\star}}{\kappa\mu r}\sqrt{\frac{\mu r}{n}}\rb\gamma^{t+1}+\frac{C_{10}}{20C_0}\frac{\sigma_r^{\star}}{\kappa}\ln\De^{t+1,\infty}\rn_{2,\infty}+\frac{\sigma_r^{\star}}{32}\ln D^{t+1,\infty}\rn_{\mathrm{F}}.
\end{align*}

Due to the bound of $\ln S^{t,m,0}\rn_{\mathrm{F}}$ in \eqref{eq:S_F}, $\tilde{B}_{2}$ has the same bound as $\tilde{B}_{1}$. For $\tilde{B}_{3}$,  
\begin{align*}
\tilde{B}_3^2 = & \ln\Ho\lb U^{t}\Si^{t} \lb D_V^{t, 0, m}\rb^T \rb V^{t+1,m}\rn_{\mathrm{F}}^2 \\
= & \left\langle \Ho\lb U^{t}\Si^{t} \lb D_V^{t, 0, m}\rb^T \rb V^{t+1,m},\Ho\lb U^{t}\Si^{t} \lb D_V^{t, 0, m}\rb^T \rb V^{t+1,m}\right\rangle \\
= & \left\langle \Ho\lb U^{t}\Si^{t} \lb D_V^{t, 0, m}\rb^T \rb ,\Ho\lb U^{t}\Si^{t} \lb D_V^{t, 0, m}\rb^T \rb V^{t+1,m}\lb V^{t+1,m}\rb^T\right\rangle \\
\leq & \ln\Ho\lb \bm{1}\bm{1}^T\rb\rn_2 \ln U^{t}\Si^{t} \rn_{2,\infty}\ln V^{t+1,m} \rn_{2,\infty}\ln D_V^{t, 0, m} \rn_{\mathrm{F}}\cdot\tilde{B}_3,
\end{align*}
where the inequality follows from  Lemma~\ref{lem:bound2}. Therefore,
\begin{align*}
\tilde{B}_3 \leq & \ln\Ho\lb \bm{1}\bm{1}^T\rb\rn_2 \ln U^{t}\Si^{t} \rn_{2,\infty}\ln D_V^{t, 0, m} \rn_{\mathrm{F}}\ln V^{t+1,m} \rn_{2,\infty}\\
\leq & c_5\sqrt{\frac{n\log n}{p}}\cdot\lb 1+\frac{5C_{10}}{C_0}\rb\sqrt{\frac{\mu r}{n}}\lb 1+\frac{1}{C_0}\rb\sigma_1^{\star}\cdot\frac{8C_{10}}{C_0}\lb\frac{1}{\kappa\mu r}\sqrt{\frac{\mu r}{n}}\rb\gamma^{t}\lb \sqrt{\frac{\mu r}{n}}+\ln \De^{t+1,\infty}\rn_{2,\infty}\rb \\
\leq & \frac{1}{2C_0}\cdot\lb\frac{\sigma_r^{\star}}{\kappa\mu r}\rb\gamma^{t+1}\lb \sqrt{\frac{\mu r}{n}}+\ln \De^{t+1,\infty}\rn_{2,\infty}\rb,
\end{align*}
provided that
$$
p\geq\frac{400 c_5^2C_{10}^2}{\gamma^2}\cdot\frac{\kappa^2\mu^2 r^2\log n}{n}.
$$

\noindent\textbf{$\bullet$ Bounding $B_4$.} When $m\leq n$, 
\begin{align*}
    B_4 = \ln e_m^T\Ho\lb L^{t,m}-L^{\star}\rb V^{t+1,m}\rn_{2}
    \leq \frac{C_{10}}{20C_0}\cdot\lb\frac{\sigma_r^{\star}}{\kappa\mu r}\rb\gamma^{t+1}\lb\sqrt{\frac{\mu r}{n}} + \ln \De^{t+1,\infty} \rn_{2,\infty}\rb,
\end{align*}
where one can see that the inequality follows from \eqref{eq:I_3}, provided that  
$$
p\geq\frac{230400C_{10}^2}{\gamma^2}\cdot\frac{\kappa^4\mu^2r^2\log n}{n}.
$$
When $m> n$,
\begin{align*}
    B_4 = & \ln \Ho\lb L^{t,m}-L^{\star}\rb e_{(m-n)}e_{(m-n)}^T  V^{t+1,m}\rn_{\mathrm{F}} \\
    \leq & \ln  \Ho\lb L^{t,m}-L^{\star}\rb \rn_2\ln V^{t+1,m}\rn_{2,\infty} \leq \frac{1}{2C_0}\cdot\lb\frac{\sigma_{r}^{\star}}{\kappa\mu r}\rb\gamma^{t+1}\lb\sqrt{\frac{\mu r}{n}} + \ln \De^{t+1,\infty} \rn_{2,\infty}\rb,
\end{align*}
where the bound of $\ln\Ho\lb L^{t,m}-L^{\star}\rb \rn_2$ from \eqref{eq:induct_op} is used to get the last inequality.

Putting the bounds of $B_1$ to $B_4$ together, one has
\begin{align*}
\ln \lb E^t-E^{t,m}\rb V^{t+1,m}\rn_{\mathrm{F}} \leq &\lb 4\times\frac{C_{10}}{20C_0}+2\times\frac{1}{2C_0}\rb\cdot\lb\frac{\sigma_{r}^{\star}}{\kappa\mu r}\sqrt{\frac{\mu r}{n}}\rb\gamma^{t+1} \\
&+ \lb 4\times\frac{C_{10}}{20C_0}+2\times\frac{1}{2C_0}\rb\frac{\sigma_r^{\star}}{\kappa}\ln \De^{t+1,\infty} \rn_{2,\infty}+\frac{\sigma_r^{\star}}{16}\ln D^{t+1,\infty}\rn_{\mathrm{F}},
\end{align*}
and consequently,
\begin{equation}\label{eq:WF_induct_bound}
\begin{aligned}
\ln W^{t,m} F^{t+1,m}\rn_{\mathrm{F}}\leq \frac{C_{10}}{2C_0}\cdot\lb\frac{\sigma_{r}^{\star}}{\kappa\mu r}\sqrt{\frac{\mu r}{n}}\rb\gamma^{t+1}+\frac{\sigma_r^{\star}}{10\kappa}\ln \De^{t+1,\infty} \rn_{2,\infty}+\frac{\sigma_r^{\star}}{8}\ln D^{t+1,\infty}\rn_{\mathrm{F}}.
\end{aligned}
\end{equation}
Thus we can finally get
\begin{align*}
\ln D^{t+1,\infty}\rn_{\mathrm{F}} \leq 2\cdot\max_m\ln D^{t+1,0,m} \rn_{\mathrm{F}}\leq\frac{2C_{10}}{C_0}\cdot\lb\frac{1}{\kappa\mu r}\sqrt{\frac{\mu r}{n}}\rb\gamma^{t+1}+\frac{2}{5\kappa}\ln \De^{t+1,\infty} \rn_{2,\infty}+\frac{1}{2}\ln D^{t+1,\infty}\rn_{\mathrm{F}},
\end{align*}
and as a result,
\begin{equation}\label{eq:D_induction}
\ln D^{t+1,\infty}\rn_{\mathrm{F}} \leq \frac{4C_{10}}{C_0}\cdot\lb\frac{1}{\kappa\mu r}\sqrt{\frac{\mu r}{n}}\rb\gamma^{t+1}+\frac{4}{5\kappa}\ln \De^{t+1,\infty} \rn_{2,\infty}.
\end{equation}

Combining \eqref{eq:D_induction} with \eqref{eq:De_induction},
\begin{equation*}
    \lb 1-\frac{4+3C_{10}}{C_0}\rb\ln \De^{t+1,\infty} \rn_{2,\infty}
    \leq \frac{21+3C_{10}}{C_0}\lb\frac{1}{\mu r}\sqrt{\frac{\mu r}{n}}\rb\gamma^{t+1}+\frac{1}{C_0}\ln \De^{t+1,\infty} \rn_{2,\infty},
\end{equation*}
\begin{equation*}
    \lb 1-\frac{5+3C_{10}}{C_0}\rb\ln \De^{t+1,\infty} \rn_{2,\infty}
    \leq \frac{21+3C_{10}}{C_0}\lb\frac{1}{\mu r}\sqrt{\frac{\mu r}{n}}\rb\gamma^{t+1}.
\end{equation*}
When $C_0\geq 5(5+3C_{10})$ and $C_{10}\geq 21$,
\begin{equation*}
    \ln \De^{t+1,\infty} \rn_{2,\infty}
    \leq \frac{5C_{10}}{C_0}\lb\frac{1}{\mu r}\sqrt{\frac{\mu r}{n}}\rb\gamma^{t+1},
\end{equation*}
substituting this bound into \eqref{eq:D_induction} and we get
$$
\ln D^{t+1,\infty}\rn_{\mathrm{F}} \leq \frac{8C_{10}}{C_0}\cdot\lb\frac{1}{\kappa\mu r}\sqrt{\frac{\mu r}{n}}\rb\gamma^{t+1}.
$$

Going through the proofs, one can see  that the induction hypotheses holds provided that
\begin{equation}\label{eq:constants}
C_{10}\geq 21,~C_0\geq 8640C_{10},
\end{equation}
and $p$ and $\alpha$ satisfy
$$
p\geq\frac{\max\left\{16c_4^2C_0^2\cdot\kappa^4\mu^3r^3,~1152c_5^2C_{10}^2 \cdot\kappa^4\mu^2r^3,~230400C_{10}^2\cdot\kappa^4\mu^2 r^2,~640000C_{10}^2\cdot\kappa^2\mu r\right\}}{\gamma^2}\cdot\frac{\log n}{n},
$$
$$
\alpha \leq \min\left\{\frac{1}{4C_0}\frac{1}{\kappa^2\mu^2r^2}\cdot\frac{\gamma}{C_{\text{thresh}}},~\frac{1}{\kappa \mu r^{3/2}}\cdot\frac{\gamma}{2000K}\right\}.
$$

\section{Conclusion} \label{sec:conclusion}

In this paper, we study the robust matrix completion problem and consider a nonconvex alternating projection method. Theoretical recovery guarantee has been established for the  method with general thresholding functions (for the sparse outlier part) based on the leave-one-out analysis, thus achieving the first projection and sample splitting free result for nonconvex robust matrix completion. For future work, we would like to consider a variant of the algorithm where the projection onto the low rank part is further accelerated via Riemannian optimization \cite{accaltprj2019,Wei2020}. In addition, it is also interesting to extend the analysis to the robust completion problem of Hankel matrices \cite{Cai2021} as well as to the robust tensor completion problem \cite{Wang2023}. 


\bibliographystyle{siam}
\bibliography{ref}

\appendix

\section{Useful Lemmas and Supplementary Proofs}

In this section we collect some useful lemmas which are used throughout this paper. Proofs are mostly provided for the new ones and when the result is different from its reference.

\subsection{Useful Lemmas}\label{subsec:useful}

\begin{lemma}\label{lem:S_op}
If $S\in\mathbb{R}^{n\times n}$ is $\alpha$-sparse, i.e., $S$ has no more than $\alpha n$ nonzero entries per row and column, then $\ln S\rn_2\leq \alpha n\cdot\ln S\rn_{\infty}$.
\end{lemma}
\begin{proof}
Denote by $\ln \cdot \rn_1$  the operator norm induced by the $l_1$ norm. The result follows from
$$
\ln S \rn_2\leq \sqrt{\ln S\rn_1\ln S^T\rn_1},
$$
and the fact that
$$
\max\left\{\ln S\rn_1, \ln S^T\rn_1\right\}\leq \alpha n\cdot\ln S\rn_{\infty}.
$$
\end{proof}

\begin{lemma}[{\cite[Lemma 2]{Chen2015}}]\label{lem:init}
Suppose $Z\in\mathbb{R}^{n\times n}$ is a fixed matrix. There exists a universal constant $c_4>1$, such that  
$$
\ln \Ho(Z)\rn_2\leq c_4\lb \frac{\log n}{p}\ln Z\rn_{\infty}+\sqrt{\frac{\log n}{p}}\cdot\max\left\{\ln Z\rn_{2,\infty},\ln Z^T\rn_{2,\infty}\right\}\rb
$$
holds with high probability.
\end{lemma}

\begin{lemma}[{\cite[Lemma 22]{Ding2020}}]\label{lem:uniform}
There exists a universal constant $c_5>1$ such that if $p \geq \frac{\log n}{n}$, then with high probability,
\begin{align*}
\left\|\mathcal{H}_{\Omega}(A)\right\|_{2} \leq c_5 \sqrt{\frac{k n \log n}{p}}\|A\|_{\infty}
\end{align*}
holds uniformly for any $A \in \mathbb{R}^{n \times n}$ satisfying $\text{rank}(A)\leq k$.
\end{lemma}

Suppose that $L=L^{\star}+E\in\mathbb{R}^{n\times n}$. Recall that the columns of $\frac{1}{\sqrt{2}}F^{\star}$ are the top-$r$ ($r<n$) orthonormal eigenvectors of $\widehat{L^{\star}}$. Denote by $\frac{1}{\sqrt{2}}F$ the top-$r$ orthonormal eigenvectors of $\widehat{L}$. Let the SVD of the matrix $H:=\frac12(F^{\star})^T F$ be $A\widetilde{\Sigma} B^T$, and define $G:=A B^T$. The following two lemmas can be readily derived from \cite[Lemma 45]{Ma2019} and \cite[Lemma 1]{Ding2020}, respectively.
\begin{lemma}[{\cite[Lemma 45]{Ma2019}}]\label{lem:op} If $\|E\|_{2}\leq\frac{1}{2}\sigma_r^{\star}$, then 
$$
\left\|F -F^{\star}G\right\|_2 \leq \frac{4}{\sigma_{r }^{\star}}\left\|E\right\|_2.
$$
\end{lemma}

\begin{lemma}[{\cite[Lemma 1]{Ding2020}}]\label{lem:perturb_gt} If $\|E\|_{2}\leq\frac{1}{2}\sigma_r^{\star}$, we have the following bounds
$$
\begin{aligned}
\left\|\Si^{\star} G-G \Si^{\star}\right\|_{2} & \leq\left(2+\frac{2 \sigma_1^{\star}}{\sigma_r^{\star}-\|E\|_{2}}\right)\|E\|_{2}, \\
\left\|\Si^{\star} H-G \Si^{\star}\right\|_{2} & \leq\left(2+\frac{\sigma_1^{\star}}{\sigma_r^{\star}-\|E\|_{2}}\right)\|E\|_{2}.
\end{aligned}
$$
\end{lemma}

Suppose that $A$ and $\widetilde{A}=A+W$ are symmetric matrices, and $\lambda_1(A)\geq \cdots\geq \lambda_r(A)>0$. Denote by  $F\La F^T$ and $\widetilde{F}\widetilde{\La}(\widetilde{F})^T$ the top-$r$ eigen-decompositions of $A$ and $\widetilde{A}$ , respectively. Let the SVD of the matrix $H:=F^T\widetilde{F}$ be $A\widetilde{\Sigma} B^T$, and define $G:=A B^T$. Define $\delta:=\lambda_r(A)-\lambda_{r+1}(A)$. As long as $\ln W\rn_2<\delta$, by the Davis-Khan theorem (see, e.g., \cite{Li1998}),
$$
\ln\sin(F,\widetilde{F})\rn \leq \frac{\ln WF\rn}{\delta-\ln W\rn_2}
$$
holds for any unitarily invariant norm. Then we have the following two lemmas, where Lemma~\ref{lem:perturb_S} is a slight modification of \cite[Lemma~2]{Ding2020}.

\begin{lemma}[{\cite[Lemma 14]{Ding2020}}]\label{lem:perturb_F}
If $\|W\|_2<\delta$, then
$$
\ln \widetilde{F}-FG\rn_{\mathrm{F}} \leq \frac{\sqrt{2}\|W F\|_{\mathrm{F}}}{\delta-\|W\|_{2}}.
$$
\end{lemma}

\begin{lemma}\label{lem:perturb_S}
If $\|W\|_{2}<\delta$, then
$$
\left\|\La G-G\widetilde{\La}\right\| \leq\left(\frac{2 \lambda_1(A)+\|W\|_2}{\delta-\|W\|_2}+1\right)\ln W F\rn,
$$
where the norm can be either the Frobenius norm or the 2-norm.
\end{lemma}
\begin{proof}
    From $\widetilde{A}\widetilde{F}=\widetilde{F}\widetilde{\La}$, we get
$\lb A+W\rb\widetilde{F} = \widetilde{F}\widetilde{\La}$. Left multiplying both sides by $F^T$ yields that
$$
\La F^T\widetilde{F} + F^TW\widetilde{F} = F^T\widetilde{F}\widetilde{\La} \Longrightarrow
\La H - H\widetilde{\La} = -F^TW\widetilde{F}.
$$
Next we only prove the result for the Frobenius norm, while the  the proof for $2$-norm is similar.
Note that \begin{align*}
    \left\|\La G-G\widetilde{\La}\right\|_{\mathrm{F}}
    \leq & \left\|\La G-\La H\right\|_{\mathrm{F}}+\left\|\La H-H\widetilde{\La}\right\|_{\mathrm{F}} + \left\|H\widetilde{\La}-G\widetilde{\La}\right\|_{\mathrm{F}} \\
    \leq & \lb \ln \La \rn_2+\ln \widetilde{\La}\rn_2 \rb \ln G-H\rn_{\mathrm{F}} + \ln F^TW\widetilde{F}\rn_{\mathrm{F}} \\
    \leq & \lb 2\lambda_1(A)+\ln W \rn_2\rb \ln G-H\rn_{\mathrm{F}} + \ln \widetilde{F}^TW^TF \rn_{\mathrm{F}} \\
    \leq & \lb 2\lambda_1(A)+\ln W \rn_2\rb \ln G-H\rn_{\mathrm{F}} + \ln W F\rn_{\mathrm{F}}, 
\end{align*}
where the last inequality follows from the fact that $W$ is symmetric. Denote the principal angles between the column space of $F$ and  $\widetilde{F}$ as $0\leq \theta_1\leq \cdots\leq \theta_r\leq \frac{\pi}{2}$, and define $\Theta = \diag(\theta_1,\cdots,\theta_r)$. Using the inequality $1 = \sqrt{\cos^2\theta + \sin^2\theta}\leq |\cos\theta|+|\sin\theta|$, we can get
\begin{align*}
\ln G-H\rn_{\mathrm{F}} = & \ln A\lb I - \widetilde{\Si}\rb B^T\rn_{\mathrm{F}} \\
= & \ln I - \widetilde{\Si} \rn_{\mathrm{F}}
= \ln I - \cos \Theta \rn_{\mathrm{F}} \leq \ln \sin \Theta\rn_{\mathrm{F}}\leq \frac{\ln WF\rn_{\mathrm{F}}}{\delta-\ln W\rn_2},
\end{align*}
where the last inequality follows from the Davis-Khan theorem.
\end{proof}

Our analysis relies heavily on the Bernstein inequality. We state a user-friendly version below, which is an immediate consequence of \cite[Theorem 1.6]{Tropp2011}.

\begin{lemma}[{\cite[Lemma 10]{Chen2014}}]\label{lem:bernstein}
Consider $m$ independent random matrices $M_l$ $(1 \leq l \leq m)$ of dimension $d_1 \times d_2$, each satisfying $\mathbb{E}\left[M_l\right]=0$ and $\left\|M_l\right\|_2 \leq B$. Define
$$
\sigma^2:=\max \left\{\left\|\sum_{l=1}^m \mathbb{E}\left[M_l M_l^T\right]\right\|_2,\left\|\sum_{l=1}^m \mathbb{E}\left[M_l^TM_l\right]\right\|_2\right\} .
$$
Then there exists a universal constant $c_{10}>0$ such that for any integer $a \geq 2$,
$$
\left\|\sum_{l=1}^m M_l\right\|_2 \leq c_{10}\left(\sqrt{a \sigma^2 \log \left(d_1+d_2\right)}+a B \log \left(d_1+d_2\right)\right)
$$
with probability at least $1-(d_1+d_2)^{-a}$.
\end{lemma}

\begin{lemma}\label{lem:bound1}
Suppose $A,B,C\in\mathbb{R}^{n\times r}$. For $1\leq m\leq n$,  define
$$
H^{(m)} = \diag\lb 1-\delta_{m1}/p,\cdots,1-\delta_{mn}/p\rb\in\mathbb{R}^{n\times n},
$$
and $R^{(m)} = B^TH^{(m)}C$. Then we have the following deterministic bound
$$
\ln e_m^T\Ho\lb AB^T\rb C\rn_2 \leq \ln e_m^T A\rn_{2}\cdot \ln R^{(m)}\rn_2.
$$
Furthermore,
$$
\ln \Ho\lb AB^T\rb C\rn_{\mathrm{F}} \leq \ln A\rn_{\mathrm{F}}\cdot \max_m\ln R^{(m)}\rn_{2}.
$$
\end{lemma}
\begin{proof}
    A direct calculation yields that 
    \begin{align*}
\ln e_m^T\Ho\lb AB^T\rb C\rn_2 = & \ln e_m^T \lb AB^T\rb H^{(m)}C\rn_{2} \\
= & \ln e_m^T A \cdot B^TH^{(m)}C\rn_{2} 
\leq \ln e_m^T A\rn_2\cdot \ln R^{(m)}\rn_2.
\end{align*}
Moreover,
\begin{align*}
\ln\Ho\lb AB^T\rb C\rn_{\mathrm{F}}^2 = & \sum_{m}\ln e_m^T\Ho\lb AB^T\rb C\rn_{2}^2 \\
\leq & \sum_{m}\ln e_m^T A\rn_2^2\cdot \ln R^{(m)}\rn_2^2 \\
\leq & \lb \max_m\ln R^{(m)}\rn_{2}^2\rb\cdot\sum_{m}\ln e_m^T A\rn_2^2
= \ln A\rn_{\mathrm{F}}^2\cdot\max_m\ln R^{(m)}\rn_{2}^2,
\end{align*}
which completes the proof.
\end{proof}

\begin{lemma}[{\cite[Lemma 8]{Chen2019}}]\label{lem:bound2} 
Suppose $A,B,C,D\in\mathbb{R}^{n\times r}$. Then we have the deterministic bound
$$
\left|\left\langle \Ho\lb AC^T\rb,BD^T\right\rangle\right| \leq \ln \Ho\lb \bm{1}\bm{1}^T\rb\rn_2 \lb \ln A\rn_{2,\infty}\ln B\rn_{\mathrm{F}}\rb\lb\ln C\rn_{\mathrm{F}}\ln D\rn_{2,\infty} \rb.
$$
\end{lemma}

The following lemma is the main tool we use to bound the outlier term in the induction steps.

\begin{lemma}\label{lem:P_Omega_AB}
Recall that $\Omega_{S^{\star}}$ is the support of the outliers in $\Omega$. If $\Po(S^{\star})$ is $2\alpha p$-sparse, then
$$
\ln \P_{\Omega_{S^{\star}}}(AB^T)\rn_{\mathrm{F}}^2\leq 2\alpha pn\cdot\min\left\{\ln A\rn_{\mathrm{F}}^2\ln B\rn_{2,\infty}^2,\ln A\rn_{2,\infty}^2\ln B\rn_{\mathrm{F}}^2\right\}
$$
holds uniformly for all $A,B\in\mathbb{R}^{n\times r}$.
\end{lemma}

\begin{proof}
In the following we prove that $\ln \P_{\Omega_{S^{\star}}}(AB^T)\rn_{\mathrm{F}}^2\leq 2\alpha pn\cdot\ln A\rn_{\mathrm{F}}^2\ln B\rn_{2,\infty}^2$. The other part can be obtained similarly. By a direct calculation, one has
\begin{align*}
    \ln \P_{\Omega_{S^{\star}}}(AB^T)\rn_{\mathrm{F}}^2
    = & \sum_{i=1}^n\sum_{(i,j)\in\Omega_{S^{\star}}} \lb A_{i,:}\lb B_{j,:} \rb^T\rb^2 \\
    \leq & \sum_{i=1}^n\sum_{(i,j)\in\Omega_{S^{\star}}} \ln A_{i,:}\rn_2^2\ln B_{j,:} \rn_2^2\\
    = & \sum_{i=1}^n\ln A_{i,:}\rn_2^2\cdot\lb \sum_{(i,j)\in\Omega_{S^{\star}}} \ln B_{j,:} \rn_2^2 \rb \\
    \leq & \sum_{i=1}^n\ln A_{i,:}\rn_2^2\cdot(2\alpha pn\ln B\rn_{2,\infty}^2)
    = 2\alpha pn\cdot\ln A\rn_{\mathrm{F}}^2\ln B\rn_{2,\infty}^2,
\end{align*}
where the second inequality follows from the sparsity assumption about $\Po(S^{\star})$.
\end{proof}

\begin{lemma}\label{lem:general_threh}
For $\lambda>0$, the soft-thresholding function $T_{\lambda}^{soft}$, and the SCAD function $T_{\lambda,a}^{scad}~(a>2)$ defined as the following
\begin{equation}\label{eq:threshold_funs}
\T_{\lambda}^{soft}(x) =
\text{sign}(x)(|x|-\lambda)_{+},~\T_{\lambda,a}^{scad}(x) = \left\{\begin{array}{cc}
\text{sign}(x)(|x|-\lambda)_{+} & |x|\leq 2\lambda \\
\frac{(a-1)x-\text{sign}(x)\cdot a\lambda}{a-2} & 2\lambda<|x|< a\lambda \\
x & |x|\geq a\lambda 
\end{array}\right.,
\end{equation}
satisfy the three properties in \textup{\labelcref{P1,P2,P3}}.
\end{lemma}

\begin{proof}
    For soft-thresholding, the first property is satisfied since $T_{\lambda}^{soft}(x)=0,~\forall\, |x|\leq \lambda$. The second property holds since it is the proximal mapping of a convex function, $|T_{\lambda}^{soft}(x)-T_{\lambda}^{soft}(y)|\leq|x-y|$. For the last property, it is easy to check that $|T_{\lambda}^{soft}(x)-x|\leq \lambda,~\forall |x|\leq\lambda$, and $|T_{\lambda}^{soft}(x)-x|= \lambda,~\forall\, |x|\geq\lambda$.

    For SCAD, it is a piece-wise linear function with knots at $\pm \lambda$, $\pm 2\lambda$, and $\pm a\lambda$. From the definition, $T_{\lambda,a}^{scad}(x)=0,~\forall |x|\leq \lambda$. By dividing $(-\infty,\infty)$ into eight intervals, and considering the slope between any two points from any two intervals, one can show that $|T_{\lambda,a}^{scad}(x)-T_{\lambda,a}^{scad}(y)|\leq \frac{a-1}{a-2}|x-y|$. Finally, it is easy to check that $|T_{\lambda,a}^{scad}(x)-x|\leq \lambda,~\forall |x|\leq\lambda$, $|T_{\lambda,a}^{scad}(x)-x|= \lambda,~\forall\, \lambda\leq|x|\leq2\lambda$, and $|T_{\lambda,a}^{scad}(x)-x|\leq \lambda,~\forall\, 2\lambda\leq|x|\leq a\lambda$. 
\end{proof}

\subsection{Proof for Lemma \ref{lem:L_infinity}}\label{sec:proof_L_infty}
Note that
\begin{align*}
    L-L^{\star} = & U\Sigma V^T-U^{\star}\Sigma^{\star}(V^{\star})^T \\
    = & (U\Sigma V^T-(U^{\star}G)\Sigma V^T)+((U^{\star}G)\Sigma V^T-(U^{\star}G)\Sigma(V^{\star}G)^T)\\
    & +((U^{\star}G)\Sigma(V^{\star}G)^T-(U^{\star}G)\Sigma^{\star}(V^{\star}G)^T)+((U^{\star}G)\Sigma^{\star}(V^{\star}G)^T-U^{\star}\Sigma^{\star}(V^{\star})^T) \\
    := & R_1+R_2+R_3+R_4.
\end{align*}
For the first three terms, we have the bounds:
$$
\ln R_1\rn_{\infty}\leq \|\Delta\|_{2, \infty}\|F\|_{2, \infty}\|\Sigma\|_{2},
~\ln R_2\rn_{\infty}\leq \|\Delta\|_{2, \infty}\|F^{\star}\|_{2, \infty}\|\Sigma\|_{2},
~\ln R_3\rn_{\infty}\leq \|F^{\star}\|_{2, \infty}^2\|E\|_{2}.
$$
For the last term,
\begin{align*}
\ln R_4 \rn_{\infty} \leq & \|F^{\star}\|_{2, \infty}^2\ln G\Sigma^{\star}G^T-\Sigma^{\star}\rn_2
= \|F^{\star}\|_{2, \infty}^2\ln G\Sigma^{\star}-\Sigma^{\star}G\rn_2,
\end{align*}
and we can invoke Lemma \ref{lem:perturb_gt} to bound the $2$-norm by considering the augmented matrices
$$
\widehat{L} = \left[\begin{array}{cc}
0               &  L^{\star}+E\\
(L^{\star}+E)^T &  0
\end{array}\right],~\widehat{L^{\star}} = \left[\begin{array}{cc}
0            & L^{\star} \\
(L^{\star})^T &  0
\end{array}\right].
$$
The top-$r$ eigen-decompositions of $\widehat{L}$ and $\widehat{L^{\star}}$ are 
$$
\lb\frac{1}{\sqrt{2}}F\rb\Sigma \lb\frac{1}{\sqrt{2}}F\rb^T~\text{and}~ \lb\frac{1}{\sqrt{2}}F^{\star}\rb\Sigma^{\star}\lb\frac{1}{\sqrt{2}}F^{\star}\rb^T,
$$
respectively. Since $\ln\widehat{L}-\widehat{L^{\star}}\rn_2=\ln\widehat{E}\rn_2=\ln E\rn_2\leq\frac12\sigma_r^{\star}$, by Lemma \ref{lem:perturb_gt}, 
$$
\ln \Sigma^{\star}G - G\Sigma^{\star}\rn_2 \leq (2+4\kappa)\ln E\rn_2.
$$

\subsection{Proof for Lemma \ref{lem:thresh}}\label{sec:proof_thresh}

For $(k,l)\in \Omega^{(-i)} \setminus \Omega_{S^{\star}}$,
$$
|(M-L^{t,i})_{kl}|=|(L^{\star}-L^{t,i})_{kl}|\leq \lb\frac{\mu r}{n}\sigma_{1}^{\star}\rb\gamma^t\leq \xi^t.
$$
Due to \labelcref{P1} of the thresholding function, $\text{Supp}\lb S^{t,i}\rb\subseteq\Omega^{(-i)}\cap\Omega_{S^{\star}}$. For $(k,l)\in \Omega^{(-i)}\cap\Omega_{S^{\star}}$,
\begin{align*}
    |(S^{t,i}-S^{\star})_{kl}|=&|\T_{\xi^t}((L^{\star}+S^{\star}-L^{t,i})_{kl})-S^{\star}_{kl}| \\
    \leq&|\T_{\xi^t}((L^{\star}+S^{\star}-L^{t,i})_{kl})-\T_{\xi^t}(S^{\star}_{kl})|+|\T_{\xi^t}(S^{\star}_{kl})-S^{\star}_{kl}|\\
    \leq&K|(L^{\star}-L^{t,i})_{kl}| + B\xi^t
    \leq (K+B)\xi^t,
\end{align*}
where in the second inequality, we use \labelcref{P2,P3} of the thresholding function. The conclusion follows by noting that $\xi^t\leq C_{\text{init}}\lb\frac{\mu r}{n}\sigma_1^{\star}\rb \gamma^t$.

\end{document}